\documentclass[11pt]{article}
   
\usepackage{xargs}   
\usepackage{times}
\usepackage{bbm}
\usepackage[pdftex,dvipsnames]{xcolor}  % Coloured text etc.

\usepackage[colorinlistoftodos,prependcaption,textsize=tiny,disable]{todonotes}
\newcommandx{\mnote}[2][1=]{\todo[linecolor=red,backgroundcolor=red!25,bordercolor=red,#1]{M: #2}}
\newcommandx{\unsure}[2][1=]{\todo[linecolor=blue,backgroundcolor=blue!25,bordercolor=blue,#1]{#2}}
\newcommandx{\knote}[2][1=]{\todo[linecolor=OliveGreen,backgroundcolor=OliveGreen!25,bordercolor=OliveGreen,#1]{#2}}
\newcommandx{\improvement}[2][1=]{\todo[linecolor=Plum,backgroundcolor=Plum!25,bordercolor=Plum,#1]{#2}}
\newcommandx{\cut}[2][1=]{\todo[linecolor=Plum,backgroundcolor=Plum!25,bordercolor=Plum,#1]{Cut: #2}}
\newcommandx{\thiswillnotshow}[2][1=]{\todo[disable,#1]{#2}}

\usepackage[bookmarks,colorlinks,breaklinks]{hyperref}  % PDF hyperlinks, with coloured links
\hypersetup{linkcolor=blue,citecolor=blue,filecolor=blue,urlcolor=blue} % all blue links

\usepackage{fullpage,appendix}
\usepackage{amsmath,amsfonts,amsthm,amssymb,xspace,bm}
\usepackage{dsfont}
\usepackage{algorithm}
\usepackage[noend]{algpseudocode}
\usepackage{tikz}
\usepackage{enumitem}

\usepackage[capitalise,nameinlink]{cleveref}
\crefname{lemma}{Lemma}{Lemmas}
\crefname{fact}{Fact}{Facts}
\crefname{theorem}{Theorem}{Theorems}
\crefname{corollary}{Corollary}{Corollaries}
\crefname{claim}{Claim}{Claims}
\crefname{example}{Example}{Examples}
\crefname{problem}{Problem}{Problems}
\crefname{definition}{Definition}{Definitions}
\crefname{assumption}{Assumption}{Assumptions}
\crefname{subsection}{Subsection}{Subsections}
\crefname{section}{Section}{Sections}

\newcommand{\newref}[2][]{\hyperref[#2]{#1~\ref*{#2}}}
\renewcommand{\eqref}[1]{\hyperref[#1]{(\ref*{#1})}}
\numberwithin{equation}{section}

\usepackage{color-edits}
\addauthor[jelani]{jn}{orange}
\addauthor[raghu]{rm}{purple}
\addauthor[tamas]{ts}{teal}
\addauthor[ravi]{rk}{olive}

\theoremstyle{plain}
\newtheorem{theorem}{Theorem}[section]
\newtheorem{lemma}[theorem]{Lemma}

\newtheorem{proposition}[theorem]{Proposition}
\newtheorem{corollary}[theorem]{Corollary}

\newtheorem{definition}[theorem]{Definition}

\theoremstyle{definition}

\newtheorem{remark}[theorem]{Remark}

\DeclareMathOperator*{\pr}{\mathsf{Pr}}

\DeclareMathOperator*{\ex}{\mathbb{E}}

\newcommand{\iprod}[2]{\langle #1,  #2 \rangle}

\def\inpw#1,#2{\langle #1, #2\rangle}

\newcommand{\reals}{\mathbb{R}}

\newcommand{\R}{\reals}

\newcommand{\one}{\mathbf{1}}
\DeclareMathOperator{\spann}{span}

\newcommand{\eps}{\epsilon}
\newcommand{\note}[1]{\marginpar{\tiny *note in TeX*}}
\newcommand{\ignore}[1]{}

\renewcommand{\phi}{\varphi}
\renewcommand{\epsilon}{\varepsilon}

\newcommand{\augmented}{\mathsf{AugmentedIndex}}
\newcommand{\countsketch}{\mathsf{CountSketch}}

\newcommand{\median}{\mathsf{median}}
\newcommand{\err}{\mathsf{err}}

\newcommand{\AR}{\textsf{AR}}
\newcommand{\AQ}{\textsf{AQ}}

\newcommand{\arank}{\mathrm{arank}}
\newcommand{\rank}{\mathrm{rank}}
\newcommand{\aQ}{\mathrm{aQ}}
\newcommand{\pq}{\mathsf{point\_query}}
\newcommand{\csk}{\countsketch}

\newcommand{\alice}{\mathrm{Alice}}
\newcommand{\bob}{\mathrm{Bob}}

\let\E\relax
\DeclareMathOperator*{\E}{\mathbb{E}}
\let\Pr\relax
\DeclareMathOperator*{\Pr}{\mathbb{P}}
\begin{document}

\author{
Jan Bul\'{a}nek\thanks{Google. \texttt{bul@google.com}.}
\and Ravi Kumar\thanks{Google Research. \texttt{ravi.k53@gmail.com}.}
\and Raghu Meka\thanks{UCLA. \texttt{raghum@cs.ucla.edu}. Supported by NSF EnCORE: Institute for Emerging CORE Methods in Data Science Award 2217033 and NSF AF: Small Award 2425350.}
 \and Jelani Nelson\thanks{Anthropic. \texttt{minilek@alum.mit.edu}.  Contributed to this work while at Google Research.}
    \and Tam\'{a}s Sarl\'{o}s\thanks{Google Research. \texttt{stamas@google.com}.}
}

\title{A Matrix Factorization Approach in Turnstile Streaming}
\date{}

\maketitle

\begin{abstract}
We define the {\it $M$-point query problem} in data streams: for a pre-defined matrix $M$, maintain a vector $x$ subject to turnstile updates such that the response $\widehat y_u$ to $\mathsf{query}(u)$ satisfies $|y_u - \widehat y_u| \le \eps\|x\|_1$, where $y := Mx$. We show that if $M$ can be factorized as $M=AB$ with $A,B$ space-efficiently representable, then a space-efficient algorithm exists, with memory $O(\eps^{-1}\|A\|_{2\rightarrow\infty}\|B\|_{1\rightarrow 1} + (\eps^{-1}\|A\|_{\infty\rightarrow\infty}\|B\|_{1\rightarrow 1})^{2/3})$ words. An important special case is when $M$ is the lower triangular matrix of all ones, which relates to the quantiles problem with additive error $\pm \eps n$, where $n$ is the database size. Our approach not only generalizes the existing ``dyadic'' approach to turnstile quantiles of Cormode and Muthukrishnan (J. Alg, 2005), but also simplifies and improves the analysis of the state-of-the-art dyadic CountSketch of Wang et al. (SIGMOD, 2013) and Luo et al. (VLDB, 2016). Our approach resembles the {\it matrix mechanism} of Li et al. (VLDB J., 2015) used in differential privacy: given a database $x\in\R^U$ and matrix $M$, the mechanism outputs a  private approximation that is close to $Mx$, where the privacy--error tradeoff depends on a certain matrix factorization norm of $M$.

We also improve the prior lower bound for quantiles with deletions, showing a memory lower bound of $\Omega(\eps^{-1}\log U)$ words. We also show any factorization has $\|A\|_{2\rightarrow\infty}\|B\|_{1\rightarrow 1} = \Omega((\log^{1.5} U) / \log\log U)$. This lower bound is new, and shows that for quantiles, the dyadic CountSketch is nearly optimal amongst factorization-based approaches.

    %There has been substantial progress recently on the problem of estimating quantiles in the streaming model, with new and/or simplified upper and lower bounds in the comparison model, word RAM model, and for relative error. For one setting however, a gap still remains between upper and lower bounds: the turnstile model, supporting both insertions and deletions.

    %We provide a new, promising matrix factorization approach to the problem. Though we do not close the gap between upper and lower bounds, we explain the existing upper bound as being an instantiation of a more general framework, while both simplifying the analysis and even slightly improving it. Specifically, we consider the $M$-point query problem:  given a matrix $M$ up front, and maintaining a vector $x\in\R^U$ subject to turnstile updates, the response to a query should be an estimate $\tilde y\in\R^U$ that well approximates $y := Mx$. Specifically, $\mathsf{query}(i)$ should return $\tilde y_i$ with $|y_i - \tilde y_i| \le \eps\|x\|_1$. We show that if $M$ has a small certain matrix factorization norm, then a space-efficient solution exists.
\end{abstract}

\thispagestyle{empty}
\addtocounter{page}{-1}
\newpage

\section{Introduction}
{\it Point query} is one of the most useful basic primitives in the design of streaming algorithms. The goal is to maintain a histogram $x\in\R^U$, initialized to the zero vector, subject to streaming updates of the form $\mathsf{update}(u, \Delta)$ that each trigger the change $x_u \leftarrow x_u + \Delta$. In the case of $x_u$ representing the multiplicity of item $u$ in a database, $\Delta=+1$ or $\Delta=-1$ can be viewed as either inserting or deleting item $u$ into or from that database, respectively. Then, an ideal point query algorithm should support the query operation $\mathsf{point\_query}(u)$, which simply returns $x_u$.

As is often the case in streaming algorithm design, such an ideal point query algorithm would require $\Omega(U)$ space, i.e., essentially memorizing the entire histogram $x$. We thus settle for an approximation in which the response to $\mathsf{point\_query}(u)$ returns  $x_u$ plus some bounded error depending on a norm of $x$.  Many low-memory algorithms have been developed that provide these guarantees for different settings (e.g., insertion-only, strict turnstile, or turnstile streaming, being deterministic or randomized), with error bounds depending on various norms of $x$ (e.g., $\ell_1$- or $\ell_2$-norm) \cite{MisraG82,CharikarCF04,CormodeM05,MetwallyAA06,BerindeICS10,NelsonNW14}. Point query has been both useful on its own, as well as a foundational primitive to enable low-memory streaming algorithms with applications in caching \cite{EinzigerF14}, heavy hitters \cite{CharikarCF04,CormodeM05}, quantiles \cite{CormodeM05,LuoWYC16}, graph streaming \cite{CormodeM05b}, and more.

\paragraph{$M$-point query problem.}

In this work, we introduce a more general problem that we call {\it $M$-point query}. Specifically, a matrix $M$ is pre-specified at the beginning of the stream, and the goal is to provide an algorithm that ideally responds to $\mathsf{point\_query}_M(u)$ with $y_u$, where we define $y := Mx$. The standard point query problem then just corresponds to $M$ being the identity matrix. As above, to achieve low memory we consider an {\it approximate} version of the problem, in which the algorithm is deemed to succeed if it outputs a value $\widehat y_u$ such that $|y_u - \widehat y_u| \leq \eps\|x\|_1$. Crucially we want error to depend \emph{only} on $\|x\|_1$ and \emph{not} on $\|Mx\|_1$, which could be much higher than $\|x\|_1$ (otherwise, we could simply apply standard point query algorithms to the vector $Mx$).

\paragraph{Rank and quantile queries in the turnstile model.}
%\paragraph{Prefix sums and rank queries.}
An important case to consider is the strict turnstile model ($x_u\ge 0$ for all $u$), where $x_u$ represents multiplicity of $u$ in a database, and with  $M = Q$ being the lower triangular matrix of all ones (aka the \emph{prefix-sum} matrix). Then $(Qx)_u = \sum_{v \le u} x_v$ is the {\it rank} of $u$ in the stream, i.e., the number of elements in the underlying database that are less than or equal to $u$. Then the $\ell_1$-error we aim for is $\pm \eps n$ with $n := \|x\|_1 = \sum_u |x_u| = \sum_u x_u$ being the database size. Note that simply using a regular point query data structure to sketch $Qx$ would give error $\pm \eps\|Qx\|_1$, which could be as big as $\Omega(\eps U\|x\|_1)$, for example if $x$ had flat entries. Obtaining error $\pm\eps\|x\|_1$ would then require parameterizing to have error $\pm \eps'\|Qx\|_1$ with $\eps' = \eps/U$. This would unacceptably increase memory usage by a factor of $U$; it would be more space-efficient to simply store $x$ in memory explicitly.

Commonly in the data structures literature one also studies the inverse query, $\mathsf{select}(j)$, which should return any item $z$ (not necessarily in the database) with $\mathsf{rank}(z) = j$. A natural relaxation given the form of additive error we consider is to return an element whose rank is not necessarily exactly $j$, but rather $j\pm\eps n$. This inverse query is often referred to as a {\it quantile} query, noting that the $\phi$th quantile is precisely $\mathsf{select}(\phi n)$, and additive error $\pm\eps n$ for $\mathsf{select}$ corresponds to error $\pm\eps$ for $\mathsf{quantile}$. We formally define these two algorithmic tasks as follows:

\begin{definition}[Approximate rank and quantile queries]\label{def:AR-AQ}
Fix parameters $\eps\in(0,1)$ and failure probability $\delta\in(0,1)$.
In a strict turnstile stream over $[U]$ producing $x\in\R_{\ge0}^{U}$ of size $n= \lceil \|x\|_1 \rceil$:
\begin{itemize}[nosep]
\item \emph{Approximate rank:} The problem $\AR(\eps,\delta,n,U)$  asks that for each $u \in [U]$, querying $\mathsf{rank}(u)=\sum_{v \le u}x_v$ returns an answer to within additive error $\pm \eps n$ with probability at least $1-\delta$.
\item \emph{Approximate quantile:} The problem $\AQ(\eps,\delta,n,U)$ asks to answer queries $q\in[n]$ by outputting an index $u\in[U]$ with $\text{rank}(u) \ge q - \varepsilon n$ and $\text{rank}(u-1) \le q + \varepsilon n$, with probability at least $1-\delta$.
\end{itemize}
For both these problems, the probability of success is for each query individually, and not for simultaneously answering all possible queries correctly.
\end{definition}

\paragraph{Motivation.}
Streaming algorithms for quantile estimation are ubiquitous in large-scale system design. When storing thousands of metrics across hundreds of microservices, a non-sketching linear memory solution may end up using more memory for the monitoring system than the actual application. Instead, each metric can be approximately tracked using a separate sketch\footnote{A {\it sketch} here denotes the memory footprint of an instance of a streaming algorithm.}. Quantile streaming algorithms have been implemented and deployed for example at Yahoo! \cite{TingMR20}, Apache Spark \cite{Spark}, Google \cite{GoogleCloud}, and more.

Algorithmic research into both upper and lower bounds in quantiles has surged in recent years. For example, we now have either matching or near-matching upper and lower bounds for insertion-only streaming quantiles in the word RAM model in which stream elements are integers in $[U]$ \cite{GuptaSW24,Wang25}, comparison-based model for deterministic algorithms \cite{GreenwaldK01,CormodeV20}, comparison-based model for randomized algorithms with very low failure probability \cite{KarninLL16}, and comparison-based model with relative error \cite{GribelyukSWY25}. These aforementioned results are all in the insertion-only setting, in which new items can be added to the database in the stream, but items may not be deleted.

In the fully dynamic setting, supporting both insertions and deletions, a gap remains. In the word RAM model, the best lower bound is $\Omega(\varepsilon^{-1})$ machine words, whereas the best upper bound is $\tilde O(\varepsilon^{-1}\log^{3/2} U)$  \cite{LuoWYC16}.%
\footnote{We use $\tilde O(f)$ to denote $O(f\cdot \mathrm{polylog}(f))$.}   Furthermore, there is only one known paradigm for producing good upper bounds: combining a tree-decomposition of the universe with point query, such as the so-called ``dyadic'' approach~\cite{CormodeM05,Wang2013Quantiles}.

%\bigskip

%Before we state our contributions, we define the concept of a factorization norm of a matrix $M$.
%\begin{definition}
%For a matrix $M$ and two matrix norms $\|\cdot\|_Y, \|\cdot\|_Z$, define the factorization norm
%$$
%\gamma_{Y,Z}(M) := \inf_{A,B: M = AB} \|A\|_Y\cdot \|B\|_Z 
%$$
%\end{definition}
%In this work we will be most concerned with the norms $\gamma^* := \gamma_{\ell_2\rightarrow\ell_\infty,\ell_1\rightarrow\ell_1}$, and $\gamma^+ := \gamma_{\ell_\infty\rightarrow\ell_\infty,\ell_1\rightarrow\ell_1}$. Recall the matrix operator norm $\|A\|_{X\rightarrow Y}$ denotes $\sup_{z\neq 0} \|Az\|_Y / \|z\|_X$ for two norms $\|\cdot\|_X, \|\cdot\|_Y$. In our case, $\gamma^*$ minimizes the product of the largest $\ell_2$-norm of a row of $A$ and $\ell_1$-norm of a column of $B$, and $\gamma^+$ minimizes the product of the largest $\ell_1$-norm of a row of $A$ and $\ell_1$-norm of a column of $B$.

%\bigskip

\subsection{Our Contributions} Our main upper bound contribution is as follows: 

\begin{theorem}\label{thm:mainintro}
Fix a matrix $M \in \R^{U \times U}$ and suppose one can factorize $M = AB$. Then, there exists a sketching algorithm that, given any $u \in [U]$, can successfully answer $\pq_M(u)$ with probability at least $2/3$ to within error $\eps \|x\|_1$, and whose memory consumption in words, beyond representing $A,B$, is 
\begin{equation}
O(1) \cdot \max\left\{\frac{\|A\|_{2\rightarrow\infty} \|B\|_{1\rightarrow 1}}{\eps}, \left(\frac{\|A\|_{\infty\rightarrow\infty} \|B\|_{1\rightarrow 1}}{\eps}\right)^{2/3} \right\} , \label{eqn:main-space}
\end{equation}
in addition to storing a random seed of $O(\log m + \log^2(1/\eps))$ bits, where $m$ is the number of rows of $B$.
\end{theorem}

We note that generally for a matrix $C$, $\|C\|_{2\rightarrow\infty}$ equals the largest $\ell_2$-norm of a row of $C$, $\|C\|_{1\rightarrow 1}$ equals the largest $\ell_1$-norm of a column, and $\|C\|_{\infty\rightarrow\infty}$ equals the largest $\ell_1$-norm of a row. The memory referenced in \eqref{eqn:main-space} does not include the memory to store $A,B$, and thus the theorem is most useful for decompositions that need not be stored explicitly but rather are implicit (such as the factorization implied by dyadic decomposition \cite{CormodeM05}).

The algorithm yielding \cref{thm:mainintro} is simple: sketch $Bx$ using $\countsketch$  with \emph{five} repetitions\footnote{Our analysis works for any constant and odd number of rows that is at least $5$.}, then estimate $y$ as $\hat{y} := A\cdot \median(\countsketch(Bx))$, where the median is applied to the five estimates; the update time depends on the column sparsity of the matrix $B$.  The analysis is novel, and departs from the typical analysis of $\countsketch$ in prior work. In particular, we need to take median of five estimates---in the standard analysis of $\countsketch$,  one repetition is enough for constant success probability---and in our analysis the random signs need not only be $O(1)$-wise independent, but also be generated by a pseudorandom generator fooling the intersection of constantly many halfspaces \cite{GopalanOWZ10,GopalanKM18,JayaramW21}. We also highlight a new \emph{correlation} property of $\countsketch$ that we believe could be useful elsewhere. \Cref{sec:overview} describes our main ideas.

\paragraph{Results for approximate quantiles.} 
We additionally provide some new contributions for the case $M=Q$, corresponding to approximate rank queries (related to quantiles as discussed):

\medskip

\noindent (i) \emph{Upper bounds.}  Via dyadic decomposition, our main upper bound yields an algorithm using $O(\eps^{-1}\log^{3/2}U + \eps^{-2/3}\log^{4/3} U) = O(\eps^{-1}\log^{3/2} U)$ words of memory, improving \cite{LuoWYC16} by a $(\log(\eps^{-1}\log U))^{3/2}$ factor and simplifying their analysis by avoiding certain conditioning and union bounds. Note that if one used the trivial decomposition $A = I$, $B = Q$, that would correspond to simply using a regular point query data structure to sketch $Qx$, which as mentioned above would require a prohibitive $\Omega(U/\eps)$ words of memory to achieve error $\pm\eps n$.

Furthermore, whereas the ``dyadic $\countsketch$'' analysis of \cite{LuoWYC16} needs an independent $\countsketch$ per level of the dyadic tree to reason about independence of errors in certain parts of the analysis, our approach implies every node in the tree decomposition can be hashed within a single $\countsketch$ instantiation, yielding a practical simplification; indeed, some industry implementations use a single table. 

\medskip 

\noindent (ii) \emph{Lower bounds.} 
We show that for approximate rank/quantiles, this factorization-based approach cannot yield a space improvement by more than $\log\log U$ factors. Specifically, we show that if $Q=AB$, then
$$
\|A\|_{2\rightarrow\infty}\|B\|_{1\rightarrow 1} = \Omega\left(\frac{\log^{3/2} U}{\log\log U}\right) .
$$
By the dyadic factorization, this is tight up to the $\log\log U$ factor. We note that some improvement to the dyadic factorization {\it is} possible; in particular, Fredman \cite{Fredman82} gave an upper bound of $\log_\lambda^{1.5} U$ with $\lambda = 3 + 2\sqrt 2 \approx 5.83$. However, our result shows that a significant asymptotic improvement is not possible.

\medskip

\noindent (iii) \emph{Separation.} We show any streaming algorithm to answer quantile queries with insertions and deletions to accuracy $\pm\eps n$ requires roughly $\Omega(\eps^{-1}\log (\epsilon n))$ words. Note this implies a separation between the fully dynamic model we consider and what is possible with insertion-only streams, where $O(\eps^{-1})$ words suffice \cite{GuptaSW24}. 

\medskip

Our matrix factorization approach resembles the {\it matrix mechanism} of \cite{LiMHMR15} in DP, and minimizing $\|A\|_{2\rightarrow\infty}\cdot\|B\|_{1\rightarrow 1}$ for $Q=AB$ is even the same as what surfaces in the privacy--utility tradeoff analysis in pure-DP for continual observation~(\Cref{app:matmech}), but the actual algorithms and analyses   differ. The recent work of \cite{BL26} proves an unconditional lower bound of $\Omega((\log U)^{1.5}/\epsilon)$ for pure-DP while the dyadic factorization (\cref{app:matmech}) gives an upper bound of $O((\log U)^2/\epsilon)$. This is reminiscent of the $\sqrt{\log U}$ gap we have in the upper and (unconditional) lower bounds for quantiles, though our factorization lower bound gives further evidence that the current upper bound may be tight. While the methods of \cite{BL26} and the proof of \cref{thm:fact-lb} have some common things (e.g., use of Haar basis), the proof ideas seem different. An intriguing open problem is to close the further $\sqrt{\log U}$ gap for pure-DP. We note the pure-DP error upper bound via the matrix mechanism involves the term $\|A\|_{2\rightarrow\infty}\|B\|_{1\rightarrow 1}\sqrt{\log U}$ (see \cref{eqn:dp-ub}), yielding $\log^2 U$ error via dyadic factorization, and thus another message of our factorization lower bound is that closing the gap to the lower bound of \cite{BL26} cannot be achieved merely by finding a better factorization.

\section{Overview of Techniques}\label{sec:overview}
\paragraph{$\countsketch$ data structure.}
We briefly recall the $\countsketch$ data structure~\cite{CharikarCF04} for sketching a vector $x \in \R^{U}$. The sketch is specified by
two parameters $\ell$ and $r$, where
$\ell$ is the number of rows (typically increased to reduce failure probability) and $r$ is the number of buckets per row (typically increased to reduce error).  The sketch consists of a hash function
$h_a:[U]\to[r]$ for each row $a \in [\ell]$, signs $\sigma_{a,u}\in\{\pm 1\}$ for each $a \in [\ell], u \in[U]$, and
$\ell \times r$ counters $C_{a,b}$ for each $a \in [\ell], b \in [r]$.   On update $(u,\Delta)$, perform
\begin{equation}\label{eq:csk1}
    C_{a,h_a(u)} \leftarrow C_{a,h_a(u)} + \sigma_{a,u}\Delta \qquad \text{for all } a\in[\ell].
\end{equation}
To answer $\pq(u)$, define the 
row estimator $\widehat{x}_u^{(a)}$
and the final estimator $\widehat{x}_u$ as
\begin{equation}\label{eq:csk2}
\widehat{x}^{(a)}_u := \sigma_{a,u} \cdot C_{a,h_a(u)}
\mbox{ and }
\widehat{x}_u := \median\big(\widehat{x}^{(1)}_u,\dots,\widehat{x}^{(\ell)}_u\big).
\end{equation}
Writing a row estimator as
$\widehat{x}^{(a)}_u = x_u + Y_{a,u}$
and expanding under full independence gives
\begin{equation}\label{eq:csk3}
Y_{a,u} := \sum_{v\neq u:\ h_a(v)=h_a(u)} \sigma_{a,u}\sigma_{a,v} x_v, \mbox{ and } \widehat{x}_u - x_u = \median(Y_{1,u},\dots,Y_{\ell,u})
=: \mathrm{err}(u).
\end{equation}
We throughout assume the signs $\sigma_{a,u}$ are fully independent; we will remove this later in \cref{sec:derandomize}. 

By \cite{LarsenPT21}, for $\ell \ge 3$ we have $\ex[\mathrm{err}(u)^2]=O(1) \cdot \|x\|_1^2/r^2$, so $r=O(1/\eps)$ gives $|\mathrm{err}(u)|\le \eps\|x\|_1$ with constant probability.

\paragraph{A new correlation bound.}
For aggregating many point queries, control of \emph{correlations} between the errors, and not just variance, is essential. We show that with $\ell=5$ rows,
\begin{theorem}[Correlation bound for median-of-5 $\countsketch$]\label{thm:cov-mains}
For any distinct $u\neq v\in[U]$,
\[
\big|\mathbb{E}[\mathrm{err}(u) \cdot \mathrm{err}(v)]\big|
= O(1/r^3)\,\|x\|_1^2.
\]
In particular, with $r=\lceil 1/\varepsilon\rceil$ this gives
\[
\big|\mathbb{E}[\mathrm{err}(u) \cdot \mathrm{err}(v)]\big|
= O(\varepsilon^3)\,\|x\|_1^2.
\]
\end{theorem}
Note that the variance bound alone, via Cauchy--Schwarz, 
would only yield
$$
\big|\mathbb{E}[\mathrm{err}(u) \cdot \mathrm{err}(v)]\big| \le \left(\mathbb E[\mathrm{err}(u)^2]\right)^{1/2}\left(\mathbb E[\mathrm{err}(v)^2]\right)^{1/2} = O(\|x\|_1^2/r^2) ,
$$
so the extra $1/r$ factor in
\cref{thm:cov-mains} is a crucial quantitative strengthening. We specialize to $\ell=5$ for clarity, but any constant $\ell\ge 5$ would suffice. We elaborate on the proof of \cref{thm:cov-mains} in \cref{sec:correlation}.
%\rmcomment{Should we have a formal statement for higher $\ell$?}\jncomment{Actually, what is the statement? And are there tradeoffs? Maybe at least worth pointing out what the tradeoffs would be.}

\paragraph{From correlation to $\pq_M$ and \Cref{thm:mainintro}.}
Let $M=AB$ be a factorization. The algorithm sketches $z:=Bx$ using $\countsketch$ (with $\ell=5$), recovers $\widehat{z}$, and outputs $\widehat{y} := A\widehat{z}$ as an estimate of $y:=Mx$. The error vector is $e:=\widehat{z}-z$, whose coordinates are exactly the point query errors, $\mathrm{err}(u)$, above. For any row $A_i$,
\[
\ex\big[\langle A_i,e\rangle^2\big]
= \sum_u A_{i,u}^2 \ex[\mathrm{err}(u)^2] + 
 \sum_{u\neq v} A_{i,u}A_{i,v}\,\ex[\mathrm{err}(u) \cdot \mathrm{err}(v)].
\]
The variance term is controlled by $\|A_i\|_2$ and the correlation term by the new $O(1/r^3)$ bound of \cref{thm:cov-mains} together with $\|A_i\|_1$. Choosing $r=\Theta(1/\eps)$ yields $|\widehat{y}_i-y_i| \le \eps\|x\|_1$ with constant probability whenever $\|A\|_{2\to\infty}\|B\|_{1\to 1}$ and $\|A\|_{\infty\to\infty}\|B\|_{1\to 1}$ are appropriately bounded, leading to the two terms in \cref{eqn:main-space} of \cref{thm:mainintro}.
Finally, to implement the sketch with limited randomness, we use constant-wise independent hashes and signs together with a PRG that fools intersections of constantly many halfspaces \cite{GopalanOWZ10,GopalanKM18,JayaramW21}, which suffices for the median analysis; details are in \cref{sec:derandomize}.

\paragraph{Quantiles: Classical factorization and upper bound.}
Let $Q \in \{0,1\}^{U\times U}$ be the lower triangular all ones (prefix-sum) matrix, so rank/quantile queries correspond to answering $\pq_Q$. The classical dyadic decomposition gives a factorization $Q=AB$ with explicit norm bounds:
\[
\|A\|_{2\to\infty}\le \sqrt{\log U},\qquad
\|A\|_{\infty\to \infty}\le \log U,\qquad
\|B\|_{1\to 1}\le \log U,
\]
and moreover $A,B$ are sparse with each row/column having at most $\log_2 U$ nonzeros. Specifically, the dyadic factorization is as follows. We imagine a perfect binary tree with leaves labeled $1,\ldots,U$ (suppose $U$ is a power of $2$). Then $B$ has $2U-1$ rows, corresponding to all nodes in the tree (including the leaves), and the $i$th column is the indicator vector of the ancestral path of leaf $i$, including $i$ itself. The $j$th row of $A$, $1\le j\le U$, is then the indicator vector of the at most $\log_2 U$ nodes, at most one per level of the tree, whose subtrees are disjoint and contain exactly the leaves $[1,j]$. $A$, $B$ can be made square via the observation that nodes that are a right child of their parent can have their rows eliminated from $B$ (and similarly columns eliminated from $A$), since their subtrees are never used in the dyadic decomposition of any prefix $[1,j]$.

Plugging the norm bounds into \Cref{thm:mainintro} yields $O(\eps^{-1}\log^{3/2} U)$ words of memory for turnstile approximate rank/quantile queries with additive error $\pm \eps n$, improving   \cite{LuoWYC16} by a $\Theta(\log^{3/2}(\eps^{-1}\log U))$ factor.  Besides this slight improvement, our framework newly implies that a single $\countsketch$ suffices across all dyadic levels: the $O(\eps^3)$ bound on the correlation lets us aggregate errors across all intervals without a separate sketch per level, simplifying the algorithm and tightening the analysis.

\paragraph{Quantiles: Factorization lower bound.}

Assume $U$ is a power of 2, $U = 2^L$ and $L \geq 10$. Let $A\in\R^{U\times m}$, $B\in\R^{m\times U}$. 
Let $a_i \in \reals^m, i \in [U]$ be the $i$th row of $A$, 
$b_k \in \reals^U, k \in [m]$ be the $k$th row of $B$, and 
$q_i = \one_{[i]} \in \reals^U$ be the $i$th row of $Q$. We refer to the $b_k$'s as the \emph{dictionary} rows.

The main idea of the proof is the following.  First we normalize the columns of $B$ so that each column
has at most unit $\ell_1$-norm, by writing $AB = (tA)(B/t)$; it then suffices to show a lower bound on the largest $\ell_2$-norm of the $a_i$'s.  To do this, we consider the depth $L$ tree defined by the dyadic intervals on $[U]$.  For every dyadic interval $I$, we consider
its left and
right children, and test $Q = AB$
against the left-minus-right Haar functional ($\phi_I$).
This leads to constraints of the form $\langle a_i,\delta_I\rangle=g_I(i)$, where  $\delta_I$ measures how imbalanced each $b_i$ is with respect to the left and right parts of the dyadic interval and $g_I$ is a function on $I$ that looks like a ``tent'' ($g_I(i) = 0$ for some prefix of $i$, then is linear with slope $-1$, then linear with slope $1$, then $0$ again). We call the inflection point where the slope changes from $-1$ to $1$ the ``kink''.

The proof then has two parts. 
\begin{itemize}[nosep]
\item Using martingale concentration, we show that $b_k$, whose dyadic imbalances $\delta_I(k)$ have large accumulated variance along many root-to-leaf paths, must contribute proportionally to its $\ell_1$-norm.  Since the column norms of $B$ have been scaled, only few such high-variance directions can occur. 
\item The kink of each tent $g_I$ implies that $\delta_I$ has a fresh component orthogonal to all same-or-coarser dyadic directions. Let $D_I$ denote the distance of $\delta_I$ to the span of the previous $\delta_{I'}$. Along any fixed root-to-leaf path, these fresh directions are orthonormal, so Bessel’s inequality forces $\|a_i\|_2^2$ to contain many terms of the fresh distances, $D_I^{-2}$. The martingale estimate bounds the total root-to-leaf sum $\sum_I D_I$ by $O(\log L)$ for most rows, while the kink property guarantees that $\Omega(L)$ levels of the tree contribute.  Hence many $D_I$’s are at most $O((\log L)/L)$, giving $\|a_i\|_2 \geq L^{3/2}/(\log L)$.  \end{itemize}
Since $L = \log U$, the lower bound follows.

\paragraph{Quantiles: Logarithmic streaming lower bound.}
We also show an unconditional lower bound for the streaming problem in \Cref{sec:loglb}: any turnstile algorithm that answers rank queries to within additive error $\pm \eps n$ with constant success probability requires
\[
\Omega\!\left(\frac{\log(\eps n)\,\log(\eps U/\log(\eps n))}{\eps}\right)\ \text{bits},
\]
which is $\Omega(\eps^{-1}\log U)$ words for $\eps$ not too small and $n$ subexponential in $U$ (and vice versa). The proof reduces from the \emph{augmented indexing} problem in the one-way communication model. In this problem, Alice receives $x\in\Sigma^t$ and Bob receives $i\in[t]$ as well as $x_1,\ldots,x_{i-1}$, and Bob must compute $x_i$ after a single message from Alice. It is known that the randomized public coin one-way communication complexity for constant success probability is $\Omega(t\log|\Sigma|)$ bits \cite{MiltersenNSW98,BarYossefJKK04}, trivially achieved by Alice simply sending $x$ entirely.

In our reduction, Alice encodes a sequence of codewords from an error-correcting code into a carefully spaced multiset stream (with geometrically decreasing block sizes), runs the streaming algorithm, and sends her  memory state to Bob. Bob deletes the prefix corresponding to his index and then issues quantile queries that recover the next codeword symbol-by-symbol. Error-correcting codes ensure robustness: even if a constant fraction of symbols are recovered incorrectly, Bob still decodes the codeword, contradicting the one-way lower bound for augmented indexing unless the sketch has the stated size. This also yields a clean separation from insertion-only streams, where $O(\eps^{-1})$ space is achievable \cite{GuptaSW24}.

\subsection{Correlation Analysis for Errors in $\countsketch$}\label{sec:correlation}

We next summarize the main argument behind \cref{thm:cov-mains}. Recall the setup of $\csk$ as defined by \cref{eq:csk1,eq:csk2,eq:csk3}. 
\ignore{
We start by recalling the definition of $\countsketch$. Fix parameters:
\begin{itemize}
  \item Number of rows $\ell=5$,
  \item Number of buckets per row $r$.
\end{itemize}

For each row $a\in[\ell]$, pick:
\begin{itemize}
  \item A hash function $h_a:[n]\to[r]$,
  \item Independent Rademacher signs $\sigma_{a,u}\in\{\pm1\}$ for each $u\in[n]$.
\end{itemize}

\paragraph{Independence assumptions.}}
As a reminder, we will first assume full independence:
different rows are independent, and within a row $a$, the random variables $\{h_a(u)\}_{u\in[U]}$ are i.i.d.\ uniform on $[r]$ and independent of the i.i.d.\ Rademacher signs $\{\sigma_{a,u}\}_{u\in[U]}$.
All the steps below can be derandomized to allow for pseudorandom $h,\sigma$ that can be stored in small memory, dominated by the rest of the algorithm (see \cref{sec:derandomize}), via limited independence and PRGs for intersections of halfspaces \cite{GopalanOWZ10}.

\ignore{
\paragraph{Sketch maintenance.}

We are trying to answer point-queries for a vector $x \in \R^n$ that is updated in a sequence of insertion/deletion operations. 

Row $a$ maintains counters $C_{a,1},\dots,C_{a,r}$ initialized to $0$.
On an update $(u,\Delta)$ (insertions/deletions), for each row $a$ perform
\[
C_{a,h_a(u)} \leftarrow C_{a,h_a(u)} + \sigma_{a,u}\Delta .
\]

\paragraph{Point query.}
For an index $i\in[n]$, define the row estimator
\[
\widehat{x}^{(a)}_i := \sigma_{a,i} \, C_{a,h_a(i)}.
\]
Then the final estimator is the median across rows:
\[
\widehat{x}_i := \median\Big(\widehat{x}^{(1)}_i,\dots,\widehat{x}^{(\ell)}_i\Big).
\]

\paragraph{Noise variables.}
Expand each row estimator as
\[
\widehat{x}^{(a)}_i
= x_i + Y_{a,i},
\qquad
Y_{a,i} := \sum_{u\neq i:\ h_a(u)=h_a(i)} \sigma_{a,i}\sigma_{a,u} x_u.
\]
Equivalently,
\[
\widehat{x}_i = x_i + \median(Y_{1,i},\dots,Y_{\ell,i}),
\qquad
\err(i) = -\median(Y_{1,i},\dots,Y_{\ell,i}).
\]}

\paragraph{Proof overview of \cref{thm:cov-mains}.} 
We use the following \emph{collision statistic}. Fix $u \neq v$ and define the number of rows in which $u$ and $v$ collide:
\[
K := \big|\{a\in[5]: h_a(u)=h_a(v)\}\big|.
\]
Since the rows are independent and $\Pr[h_a(u)=h_a(v)]=1/r$, we have $K\sim\mathrm{Bin}(5,1/r)$. We write
\[
\E[\err(u) \cdot \err(v)]
=
\sum_{k\ge 0} \E[\err(u) \cdot \err(v) \cdot \mathbf{1}_{\{K=k\}}].
\]
We show:
\begin{itemize}[nosep]
  \item If $K=0$, then $\E[\err(u) \cdot \err(v)\mid K=0]=0$ by a sign-flip symmetry argument.
  \item If $K\ge 3$, then the contribution is bounded by $\|x\|_1^2\Pr[K\ge 3] \le 10\|x\|_1^2/r^3$.
  \item If $K=1$, then a median robustness argument yields
    $\E[\err(u)^2\mid K=1] =  O(\|x\|_1^2/r^2)$,
    hence the covariance contribution is $O(\|x\|_1^2/r^3)$.%\footnote{We use $A\lesssim B$ to denote $A = O(B)$.}
  \item If $K=2$, then a simpler robustness argument yields
    $\E[\err(u)^2\mid K=2]=  O(\|x\|_2^2/r) \le O(\|x\|_1^2/r)$,
    hence the covariance contribution is also $O(\|x\|_1^2/r^3)$.
\end{itemize}
In the interest of space, we describe the details of each of the steps in \cref{sec:correlationfull}. 

\begin{theorem}\label{th:cskderand}
    We can choose pseudorandom hash functions $h_a:[U] \rightarrow [r]$, and signs $\sigma_{a,u} \in \{\pm 1\}$ with a total seed of length $O(\log U + \log^2(1/\varepsilon))$ bits to satisfy the same guarantees as in \cref{thm:cov-mains}. In addition, the signs and hash functions can be computed efficiently. 
\end{theorem}

\paragraph{AI acknowledgments.} 
GPT5.5-Pro and Gemini 3.1 Pro were used in checking the calculations and polishing parts of the paper. All of the main ideas, except for the lower bound for $\gamma_{2,1}(Q)$ in \cref{sec:fac-lb} are mainly  due to the authors. 

While public frontier models (e.g., GPT5.5-Pro, Anthropic Fable 5) failed to make meaningful progress on showing lower bounds on $\gamma_{2,1}(Q)$, the present proof of this lower bound was generated entirely by an internal version of the UCLA AI for Math Moonshot harness \cite{uclamoonshot}. The proof generated by the harness was rewritten by the authors here for clarity and consistency with the rest of the writeup.

\section{Using Factorizations for $\pq_M$}

We now prove \cref{thm:mainintro}, which states that if certain operator norms of the factorization of $M$ are small, then a space-efficient algorithm exists.  

\begin{proof}[Proof of \cref{thm:mainintro}]
Let $M=AB$ be a factorization. The algorithm sketches $z:=Bx$ using $\countsketch$ with $\ell = 5$, and a parameter $r$ to be chosen later. We recover $\widehat{z}$ as given in \cref{eq:csk1,eq:csk2}, and output $\widehat{y} := A\widehat{z}$ as an estimate of $y:=Mx$. The error vector is $e:=\widehat{z}-z$, whose coordinates are exactly the point query errors $\mathrm{err}(u)$ as defined in \cref{eq:csk3} for $\csk$.

Our goal is to bound $\ex[(y_u - \widehat{y}_u)^2]$ for any $u \in [U]$. Observe that $y_u - \widehat{y}_u = \iprod{A_u}{e}$, where $A_u$ is the $u$th row of $A$. Now, 
\[
\ex\big[\langle A_u,e\rangle^2\big]
= \sum_v A_{u,v}^2 \ex[\mathrm{err}(v)^2] + 
 \sum_{v'\neq v} A_{u,v}A_{u,v'}\,\ex[\mathrm{err}(v) \cdot \mathrm{err}(v')].
\]

By \cite{LarsenPT21}, we know that for any $v$, $\ex[\mathrm{err}(v)^2] = O(1/r^2) \|z\|_1^2$. From \cref{thm:cov-mains}, we also have that for $v \neq v'$, $|\ex[\mathrm{err}(v)  \cdot \mathrm{err}(v')]| = O(1/r^3) \|z\|_1^2$. Also, note that $\|z\|_1 = \|B x\|_1 \leq \|B\|_{1 \rightarrow 1} \|x\|_1$. 

Combining these equations, we get 
\begin{align*}
    \ex[\langle A_u, e\rangle^2] &\leq O(1/r^2) \cdot \|A_u\|_2^2 \cdot \|B\|_{1 \to 1}^2 \|x\|_1^2 + O(1/r^3) \cdot \|A_u\|_1^2 \cdot  \|B\|_{1 \to 1}^2 \|x\|_1^2\\
    &\leq O(1/r^2) \cdot \|A\|_{2 \to \infty}^2 \cdot \|B\|_{1 \to 1}^2 \|x\|_1^2 + O(1/r^3) \cdot \|A\|_{\infty \to \infty}^2 \cdot  \|B\|_{1 \to 1}^2 \|x\|_1^2,
\end{align*}
where the last inequality follows from the definition of $\|\cdot\|_{2 \to \infty}$ and $\|\cdot\|_{\infty \to \infty}$ (they are the max $\ell_2$-norm and max $\ell_1$-norm of any row of $A$, respectively). 

Thus, to apply Markov's inequality and have the final expression be at most $\eps^2 \|x\|_1^2/3$, it suffices to take 

$$r = O(1)\cdot \max\left(\frac{\|A\|_{2 \to \infty} \cdot \|B\|_{1 \to 1}}{\eps}, \left(\frac{\|A\|_{\infty \to \infty} \cdot \|B\|_{1 \to 1}}{\eps}\right)^{2/3}\right) . $$

Finally, to get the memory bound, we use the derandomized version of $\csk$ as given by \cref{th:cskderand}, applied to the vector $Bx\in\R^m$ (where $B$ has $m$ rows), finishing the proof. 
\end{proof}

\ignore{
\section{Application: Turnstile Rank and Quantile Queries}\label{sec:quantiles}

\paragraph{Prefix sums and rank queries.}
For strict turnstile streams, the frequency vector $x\in\R_{\ge 0}^{U}$ induces a multiset of size
$n:=\|x\|_1$ over $[U]$.
Define the \emph{prefix-sum matrix} $Q\in\{0,1\}^{U\times U}$ by
\[
Q_{i,j} := \ind{j\le i}.
\]
Then $(Qx)_i=\sum_{j\le i} x_j$ equals the rank of $i$ in the multiset.
Thus, the $(Q$-point query problem) is exactly approximate rank queries with additive error $\pm \eps n$.

Approximate quantile queries (given $q\in[n]$, output an element of rank in $[q-\eps n,q+\eps n]$)
reduce to approximate rank queries via binary search over $[U]$.

\rmcomment{Do we need the next notation. It makes things simpler but it is being introduced kind of late.}

\begin{definition}[Approximate rank and quantile queries]\label{def:AR-AQ}
Fix parameters $\eps\in(0,1)$ and failure probability $\delta\in(0,1)$.
In a strict turnstile stream over $[U]$ producing $x\in\R_{\ge0}^{U}$ of size $n=\|x\|_1$:
\begin{itemize}
\item The problem $\AR(\eps,\delta,n,U)$ (approximate rank) asks to answer queries $\mathsf{rank}(i)=\sum_{j\le i}x_j$ to additive error $\pm \eps n$ with probability at least $1-\delta$.
\item The problem $\AQ(\eps,\delta,n,U)$ (approximate quantile) asks to answer queries $q\in[n]$ by outputting an index $i\in[U]$ whose rank lies in $[q-\eps n,q+\eps n]$ with probability at least $1-\delta$.
\end{itemize}
\end{definition}}

\section{Quantiles and the Dyadic Factorization}

The classical dyadic decomposition yields an explicit factorization of $Q$, which we restate for completeness.

\begin{lemma}[Dyadic factorization]\label{lm:dyadicfact}
Assume $U$ is a power of two. There is a factorization $Q=AB$ with $d=2U-1$ columns/rows such that
\[
\|A\|_{2\to\infty}\le \sqrt{\log_2 U+1},\qquad
\|A\|_{\infty\to\infty}\le \log_2 U+1,\qquad
\|B\|_{1\to 1}\le \log_2 U+1.
\]
Moreover, $A$ is row-sparse and $B$ is column-sparse: every row of $A$ and every column of $B$ has at most $\log_2 U+1$ nonzeros.
\end{lemma}

\begin{proof}
Consider the perfect binary tree on $U$, with the leaves corresponding to the dyadic intervals of $[U]$.
Let $d=2U-1$ be the number of nodes.
Index the rows of $B\in\{0,1\}^{d\times U}$ by tree nodes and the columns by leaves: $B_{v,j}=1$ iff node $v$ lies on the root-to-leaf path of leaf $j$.
Then each column has exactly $\log_2 U+1$ ones, so $\|B\|_{1\to 1}\le \log_2 U+1$.

Define $A\in\{0,1\}^{U\times d}$ by letting row $i$ be the indicator of the (at most $\log_2 U+1$) dyadic intervals whose disjoint union is the prefix $[1,i]$.
Then each row has at most $\log_2 U+1$ ones, hence $\|A\|_{\infty\to\infty}\le \log_2 U+1$ and $\|A\|_{2\to\infty}\le \sqrt{\log_2 U+1}$.
By construction, $AB=Q$.
\end{proof}

%\subsection{Space bound for turnstile rank queries}

\begin{theorem}
    There is an $O((\log ^{3/2} U)(\log(1/\delta))/\eps)$-word turnstile algorithm for $\AR(\eps,\delta,n,U)$. Similarly, there is an $O((\log^{3/2} U)(\log(1/\eps \delta))/\eps + \log \log U)$-word turnstile algorithm for $\AQ(\eps,\delta,n,U)$.
\end{theorem}
\begin{proof} 

Plugging \cref{lm:dyadicfact} into \cref{thm:mainintro} yields an $O(\log^{3/2}U/\eps)$-word algorithm for turnstile approximate rank queries.
More precisely, since
\[
\|A\|_{2\to\infty}\|B\|_{1\to1} = O(\log^{3/2} U)
\quad\text{and}\quad
\|A\|_{\infty\to\infty}\|B\|_{1\to1} = O(\log^{2} U),
\]
\cref{thm:mainintro} gives a space bound (in words) of
\[
O\!\left(
\max\left\{
\frac{\log^{3/2}U}{\eps},\ 
\left(\frac{\log^{2}U}{\eps}\right)^{2/3}
\right\}
\right)\ 
=
O\!\left(\frac{\log^{3/2}U}{\eps}\right)\,,
\]
where the final simplification holds for all $\eps\le 1$.

This gives an algorithm for solving $\AR(\eps,1/3,n,U)$.  Note that since $B$ has at most $\log_2 U + 1$ non-zeros per column, the update time 
is $O(\log U)$.  To boost the success probability, we can do the standard trick of instantiating $O(\log(1/\delta))$ independent copies and taking a median of the estimators. 

To obtain an algorithm for solving $\AQ(\eps,\delta,n,U)$ we use the standard reduction that makes multiple adaptive  (typically $O((1/\eps) \log U)$) queries to an $\AR$ oracle~\cite{CormodeM05, LuoWYC16}.
\end{proof}

As mentioned earlier, focusing on $\delta = 1/3$, the above result has two advantages over \cite{LuoWYC16}: (i) it simplifies the analysis and algorithm requiring a single $\csk$ table for the entire tree as is actually used by practitioners and (ii) saves an additional $O((\log(\eps^{-1} (\log U)))^{3/2})$ factor in space. 

\subsection{Further Applications}

We can obtain dyadic-like factorizations from \cref{lm:dyadicfact} for matrices $M$ that have a specific \emph{tree-like} or small \emph{treewidth} structure. Concretely, the arguments of \cite{KelnerKMR21} show that whenever the graph formed by the support of entries in $(M M^T)^{-1}$ has treewidth $k$, there is a factorization of $M = A B$ where the rows of $A$ and columns of $B$ are $O(k \log U)$ sparse. This effectively leads to $\|A\|_{2 \to \infty} \|B\|_{1 \to 1} \approx O((k \log U)^{3/2})$ as long as the coefficient sizes are bounded which is typically the case. Thus, for all such matrices $M$, we get an algorithm for $\pq_M$ that uses $O((k \log U)^{3/2}/\eps)$ words. 

The above includes matrices $M$ that generalize quantiles to allow for equations like: $y := Mx$, with $y_i = \sum_{j=1}^k a_{j} y_{i-j} + x_i$. Note that when $y := Q x$, we have $y_i = y_{i-1} + x_i$ and is thus a special case of the above. 
This can also capture \emph{exponentially-decaying} sums: $y_i = (1-\rho) y_{i-1} + x_i$, which corresponds to the matrix $M_\rho$ where $M_\rho(i,j) = (1-\rho)^{i-j}$ if $j \leq i$ and $0$ otherwise. We would still get a turnstile algorithm with same complexity as for quantiles by the same method. 

\section{Logarithmic Lower Bound}
\label{sec:loglb}

Here we prove that any turnstile streaming algorithm that approximates rank queries to within error $\eps$ while supporting deletions must use roughly $\Omega(\eps^{-1} \log (\eps n))$ words. Note that this is logarithmically worse than what can be done in the insertion-only setting.%\jncomment{This is a lower bound for $\AQ$, whereas our upper bound is for $\AR$ in rest of paper; should make sure to spell out the connection}

\begin{theorem}\label{th:mainlb}
There is a fixed constant $\eps_0$ such that the following holds. Suppose streaming algorithm $\mathcal A$ solves the $\AQ(\eps,\delta,n,U)$ problem in strict turnstile streams for some $\eps\in(0,\eps_0)$ and $\delta<1/10$. Then $\mathcal A$ must use at least $\Omega(\eps^{-1}\log ( \eps n)\log(\eps U/\log(\eps n)))$ bits of memory. 
\end{theorem}

The following lemma is standard.

%% Proof below is the  https://en.wikipedia.org/wiki/Gilbert%E2%80%93Varshamov_bound_for_linear_codes, we can also/additionally cite, no deterministic constructions are known.

\begin{lemma}\label{lem:code-exists}
For any integers $q,\ell>1$ and $\rho \le q/(2e)$, there exists a code $\mathcal C_{q,\ell,\rho}$ with $|\mathcal C_{q,\ell,\rho}| \ge (q/(e\rho))^{\ell/(2\rho)}$, alphabet size $q$, block length $\ell$, and relative distance $1 - 1/\rho$.
\end{lemma}
\begin{proof}
   Pick $C_1,\ldots,C_N$ independently and uniformly at random from $[q]^\ell$. Define $\alpha(C,C') := |\{ k : C_k = C'_k \}|$. Then for any $i\neq j$, $\ex \left[ \alpha(C_i, C_j) \right] = \ell/q$, so that by a Chernoff bound
    $$
    \Pr \left[\alpha(C_i, C_j) > \frac{\ell}{\rho} \right] < \left(\frac {e\rho}{q} \right)^{\frac \ell\rho} .
    $$
    By a union bound, $\Pr[\exists i\neq j,\ \alpha(C_i, C_j) > \ell/\rho] < N^2 ((e\rho)/q)^{\ell/\rho} \le 1$ for $N = (q/(e\rho))^{\ell/(2\rho)}$. Thus, a code of size $N$ with the desired property exists.
\end{proof}

% We also make use of Fano's inequality.

% \begin{lemma}[Fano's inequality]
% Suppose $X$ is a random variable finitely supported in $\mathcal{X}$. Let $\hat{X} \eqdef g(Y)$ be the predicted value of $X$ for $g$ a deterministic function also taking values in $\mathcal{X}$. Then if $\Pr[\hat{X}\neq X] = \delta$,
% $$
% H(X|Y) \le H(X|\hat{X})\le H_2(\delta) + \delta\log_2(|\mathcal{X}|-1), 
% $$
% where $H_2(\delta) \eqdef -\delta\log_2 \delta - (1-\delta)\log_2(1-\delta)$.
% \end{lemma}

\begin{definition}[Augmented indexing]
In the $\augmented_{n,\Sigma}$ problem, there are two parties Alice and Bob. Alice receives $x\in \Sigma^n$, and Bob receives $x_1,\ldots,x_{i-1},i$. Alice must send a single message to Bob, who must then output some $b\in \Sigma$. We would like that $\Pr[b\neq x_i] \le \delta$ for some given $\delta\in(0,1)$.
    \end{definition}

    \begin{theorem}[\cite{BarYossefJKK04,MiltersenNSW98}]
    The one-way randomized communication complexity of $\augmented_{n,\{0,1\}}$ with public randomness and failure probability $\delta < 1/3$ is $\Omega(n)$. Succinctly,
    $$
    \forall \delta < 1/3,\ R_\delta^{\text{pub},\rightarrow}(\augmented_{n,\{0,1\}}) = \Omega(n) .
    $$
    \end{theorem}

    \begin{corollary}\label{cor:augindex}
    Suppose $|\Sigma|\ge 2$ is a power of $2$. Then
    $$
    \forall \delta < 1/3,\ \R_\delta^{\text{pub},\rightarrow}(\augmented_{n,\Sigma}) = \Omega(n\log|\Sigma|) .
    $$
    \end{corollary}
    \begin{proof}
    We reduce from $\augmented_{n\log_2(|\Sigma|),\{0,1\}}$. Let $x\in\{0,1\}^{n\log_2(|\Sigma|)}$ be Alice's input, and $x_1, \ldots, x_{i-1}, i$ be Bob's input. Let $\Pi$ be a protocol solving $\augmented_{n,\Sigma}$ involving players $\Pi_{\alice}$ 
    and $\Pi_{\bob}$. Alice breaks up her input into $n$ blocks of $\log_2(|\Sigma|)$ bits to then create a new input $x'\in\Sigma^n$ by treating each block as an element of $\Sigma$. She then simulates  $\Pi_{\alice}$ with input $x'$ and sends the corresponding message. Bob similarly simulates $\Pi_{\bob}$ to recover not just the $i$th bit of $x$, but the entire $\log_2(|\Sigma|)$-bit block of $x$ containing $x_i$.
    \end{proof}

We are now ready to prove a lower bound for the approximate quantiles problem when deletions are allowed in the stream.

\begin{proof}[Proof of \cref{th:mainlb}]
    Without loss of generality we assume $\mathcal A$ has error parameter $\eps/10$. We reduce from $\augmented_{s,\mathcal C_{q,\ell,\rho}}$ for $q = \eps U/\log(\eps n), \rho = 100, \ell = K/\eps, s = (1/2)\log_{10}(\eps n)$, and $K>0$ a sufficiently small constant to be determined, where ${\mathcal C}_{q,\ell,\rho}$ is from \cref{lem:code-exists} (we can remove some of its elements to round down its size to a power of $2$). Alice's input is a sequence  $C_1,\ldots,C_s\in \mathcal C_{q, \ell,\rho} \subset [q]^\ell$. Bob receives $C_1,\ldots,C_{i-1},i$ for some $i\in[s]$ and would like to compute $C_i$.  Let $\mathcal A$ be a streaming algorithm for the approximate quantiles problem.
    
    To use $\mathcal A$ to solve $\augmented$, Alice creates a stream from her input as follows. For each $i\in[s]$ and $j\in[\ell]$, she inserts $((i-1)\ell + (j-1))q + (C_i)_j$ into the stream $\eps n/11^i$ times.  (By construction, the intervals of values corresponding to any $(i, j)$ are disjoint.) She then takes the memory contents $\mathcal M$ of $\mathcal A$ and sends that to Bob as a message. Bob then initializes $\mathcal A$ in memory state $\mathcal M$. He then deletes items inserted into the stream by Alice due to $C_{i'}$ for $i' < i$, then immediately afterward stores a snapshot of the memory state $\mathcal M'$ of $\mathcal A$; he also defines $n' := \sum_{k=i}^s \eps n\ell/11^k < 1.1Kn/11^i$. 
    
    We now describe how Bob attempts to recover each symbol of $C_i\in[q]^\ell$, one index at a time. More specifically, he will compute some $\tilde C\in[q]^\ell$, which he will then attempt to use to recover $C_i$. Bob loops over each $j\in[\ell]$, and in each iteration does the following. He first resets the memory of $\mathcal A$ to $\mathcal M'$. He then queries $\mathcal A$ for an item $z_j\in[U]$ with rank $t = \lfloor(j-1/2)\eps n/11^i\rfloor$, then sets $\tilde C_j := z_j - ((i-1)\ell + (j-1))q$.

    Note that the elements of the universe that might be generated in the stream due to the symbol at $(C_i)_j$ are all strictly less than those that may be generated by the symbol at $(C_{i'})_{j'}$ for $(i,j) < (i',j')$ (lexicographic comparison). Furthermore, since the symbol at $(C_i)_j$ causes $\eps n/11^i$ insertions in the stream, from the point in the stream corresponding to memory state $\mathcal M'$, the item $\beta := ((i-1)\ell + (j-1))q + (C_i)_j$ spans an interval of ranks $[(j-1)\eps n/11^i, j\eps n/11^i)$, and Bob queries $t$ precisely in the middle of this interval. Thus, $\mathcal A$ will return $\beta$ as long as its rank error is less than $
    \Delta := (\eps n)/(2\cdot 11^i)$. By the guarantees of an $\AQ$ data structure, $\mathcal A$ will have rank error at most $(\eps/10) n'< .11\eps n/11^i < \Delta$ with probability at least $1-\delta$. Thus, by the linearity of expectation, Bob expects that $\tilde C_j \neq (C_i)_j$ for at most $\delta \ell < \ell/10$ values of $j\in[\ell]$. Then by Markov's inequality, the probability that $\tilde C_j \neq (C_i)_j$ for more than $\ell/3$ values of $j$ is less than $1/3$. Conditioned on this bad event not occurring, Bob can uniquely correct $\tilde C$ to $C_i$ using the fact that $\mathcal C_{q,\ell,\rho}$ is an error-correcting code of relative distance $1-1/\rho = 99/100$, thus solving the $\augmented$ instance.
 
    The number of bits communicated in this protocol is equal to the space complexity of $\mathcal A$, which by \cref{cor:augindex} is $\Omega(s\log|\mathcal C_{q,\ell,\rho}|) = \Omega(\ell s\log q) = \Omega(\eps^{-1}\log(\eps U/\log(\eps n))\log(\eps n))$, as desired.
\end{proof}

\cref{th:mainlb} is an improvement over the insertion-only lower bound of roughly $\Omega(1/\eps)$ words (see \cite[Theorem 3.1.9]{Nelson20}).
\section{Factorization Lower Bound}\label{sec:fac-lb}
In this section we prove the lower bound on $\gamma_{2,1}(Q) := \inf_{Q = AB}\|A\|_{2\rightarrow\infty}B_{1\rightarrow 1}$.

\begin{theorem}\label{thm:fact-lb}
$\gamma_{2,1}(Q) = \Omega((\log^{1.5} U)/(\log\log U))$.
\end{theorem}

First we introduce some preliminaries.

\subsection{Preliminaries}
The following lemma essentially says that the absolute value function cannot be well approximated by an affine function. 
\begin{lemma}\label{lem:tentaffine}
For an integer $N$, let $h:\{-N,\ldots,N\} \rightarrow \R$ be defined by $h(x) = -\frac{1}{2} + \frac{|x|}{N}$. Then, for any affine function $P:\R \rightarrow \R$, 
$$\left|\left\{i \in \{-N, \ldots, N\}: |h(i) - P(i)| \geq \frac{1}{16}\right\}\right| \geq  \frac{7N}{8}.$$
\end{lemma}
\begin{proof}
Fix a constant $\alpha > 0$. 
Let $G = \{i: |h(i) - P(i)| < \alpha\}$; our goal is to lower bound $|\overline{G}|$. Call an index $i$ \emph{doubly-good} if $i \in G$ and $-i \in G$; let $D \subseteq G$ be the set of doubly-good indices.  Note that 
$2N+1 - |\overline{G}| = |G| = |D| + |G \setminus D| \leq |D| + |\overline{G}|$, from which $|\overline{G}| \geq N + (1 - |D|)/2$.

Note that, as $P$ is affine we have $P(i) + P(-i) = 2 P(0)$. Thus, for $i \in D$, since $h$ is symmetric, we have $h(i) \in [P(0) - \alpha, P(0) + \alpha] $. But the number of indices $i$ such that $h(i)$ lies in an interval of width $2\alpha$ is at most $4 N \alpha + 1$; thus, $|D| \leq 4 N \alpha + 1$.  
The proof now follows by setting $\alpha=1/16$. 
\end{proof}

\begin{lemma}\label{lem:innerprod}
Let $L \ge 1$, and let $X,Y \in \mathbb{R}_+^m$ be random vectors such that, almost surely,
\[
\|X\|_\infty \le 1,
\qquad
\|X\|_2^2 \le L,
\qquad
\|Y\|_\infty^2 \le L.
\]
Suppose in addition, we have an empirical tail estimate that, for every $t>0$,
\[
\mathbb{E} [\bigl|\{j \mid Y_j>t\}\bigr|]
\lesssim \frac{1}{t}.
\]
Then
\[
\mathbb{E} [\langle X,Y\rangle] 
\lesssim \log(2+L).
\]
\end{lemma}

\begin{proof}
For $t>0$, let $N(t):=\bigl|\{j \mid Y_j>t\}\bigr|$. Every coordinate of $Y$ is at most $\sqrt{L}$. Thus, by the layer-cake formula,
\[
\mathbb{E}\langle X,Y\rangle
=
\int_0^{\sqrt{L}}
\mathbb{E}\left[\sum_{j:Y_j>t}X_j\right]\,dt.
\]
For every $t>0$, the bound $\|X\|_\infty\le 1$ gives
\[
\sum_{j:Y_j>t}X_j \le N(t),
\]
while Cauchy--Schwarz gives
\[
\sum_{j:Y_j>t}X_j
\le
\|X\|_2\sqrt{N(t)}
\le
\sqrt{L\,N(t)}.
\]
Taking expectations, using the assumed tail bound, and applying Jensen's inequality, we obtain
\[
\mathbb{E}\left[\sum_{j:Y_j>t}X_j\right]
\lesssim
\min\left\{\frac{1}{t},\sqrt{\frac{L}{t}}\right\}.
\]

The two bounds are equal when $t=1/L$. Therefore,
\[
\begin{aligned}
\mathbb{E} [\langle X,Y\rangle] 
&\lesssim
\int_0^{1/L}\sqrt{\frac{L}{t}}\,dt
+
\int_{1/L}^{\sqrt{L}}\frac{dt}{t} 
\lesssim
1+\log L
\lesssim
\log(2+L).  \qedhere
\end{aligned}
\]
\end{proof}

Finally, we need the following elementary inequality:
\begin{lemma}\label{lem:jensen}
For a positive vector $v \in \R^m$, if $\sum_i v_i \leq K$, then 
$$ \sum_{i=1}^m \frac{1}{v_i^2} \geq \frac{m^3}{K^2}.$$
\end{lemma}
\begin{proof}
    Applying Jensen's inequality, 

$$\frac{1}{m} \sum_{i=1}^m \frac{1}{v_i^2} \geq \left(\frac{\sum_{i=1}^m v_i}{m}\right)^{-2}
\geq \frac{m^2}{K^2}.  \qedhere$$
\end{proof}

%The rest of this section proves \cref{thm:fact-lb}.

\subsection{Dyadic tree, tents, and its properties}

We are now ready to go into the main proof. First, by scaling $A$ and $B$ suitably, 
we can assume that the $\ell_1$-norm of each column of $B$ is at most 1:
\begin{equation}
\label{eq:ink}
\sum_{k = 1}^m |b_k(j)| \leq 1,
\mbox{ for all } j \in [U] ,
\end{equation}
where we treat the $k$th row of $B$ as a function $b_k:[U]\rightarrow\reals$.
Hence it suffices to show a lower bound on
$\max_i \|a_i\|_2$, where $a_i$ is the $i$th row of $A$.

We consider the complete dyadic tree over the leaves $[U]$.  For $1\le \ell \le L$, a \emph{scale-$\ell$ node} with offset $b$
in this tree is the dyadic
interval $I=\{b+1,\dots,b+2^\ell\}$ with $b \equiv 0 \pmod{2^\ell}$; let $I_L, I_R$ be the left/right halves of $I$.  Let $I_\ell(j)$ denote the scale-$\ell$ node containing $j$; by convention $I_0(j) = \{j\}$. Define the ``Haar functional''
\[
\phi_I := \frac{\one_{I_R}-\one_{I_L}}{|I|},\qquad \delta_I\in\R^m,\ \ \delta_I(k):=\langle b_k,\phi_I\rangle,\qquad g_I(i):=\langle q_i,\phi_I\rangle.
\]
Intuitively, $\phi_I$ is a signed test function that compares the right half of $I$ to its left half; $\delta_I(k)$ measures how much the row $b_k$ is biased toward the right half of $I$ versus its left half.

The quantity $g_I$ is the same left-right test applied to the prefix
vector $q_i$.  Let us examine the structure of the function $g_I$.  Note that for every row $i$ and every node $I$:
\begin{equation}\label{eq:I1}
\langle a_i,\delta_I\rangle \;=\;\Big\langle \sum_k a_i(k)b_k,\ \phi_I\Big\rangle\;=\;\langle q_i,\phi_I\rangle\;=\;g_I(i).
\end{equation}
It is easy to see the following: $g_I(i)=0$ for $i\notin I$ (indeed, if $i< b+1$ then $q_i$ is disjoint from $I$; if $i\ge b+2^\ell$ then $q_i$ covers $I$ and $\langle \one_I,\phi_I\rangle=0$), and for $i=b+x$, $1\le x\le 2^\ell$,
\[
g_I(b+x) \;=\; -\frac{\min(x,\,2^\ell-x)}{2^\ell},
\]
a ``tent'' of depth $1/2$ at the midpoint $x=2^{\ell-1}$, with value $0$ at $x=2^\ell$.  

For convenience, we extend each $g_I$  to a continuous piecewise-affine function of a real variable in the natural canonical way: zero outside $I$, decreases linearly on the left half of $I$, reaches depth $-1/2$ at the midpoint, and then increases linearly back to 0 on the right half of $I$.
Its only ``kinks'' are at $b,\ b+2^{\ell-1},\ b+2^\ell$---all multiples of $2^{\ell-1}$. Whenever we say ``$g_I$ is affine on an interval,'' we refer to this canonical extension.

We will use the fact that the $V$-shaped tent functions $g_I$ cannot be well approximated by a straight line on most of its domain.  This follows immediately from \cref{lem:tentaffine}, using 
$N = |I|/2$ for an interval $I$.
\begin{lemma}\label{lem:kink}
Let $I$ be a node at scale $\ell\ge 6$, offset $b$. Let $P:I\to\R$ be the restriction to $I$ of an affine function on the real interval $[b,b+2^\ell]$. Then
\[
\#\Big\{i\in I:\ |g_I(i)-P(i)|\ge \tfrac1{16}\Big\}\;\ge\;\frac{7|I|}{16}.
\]
\end{lemma}

\subsection{Mean martingale and its properties}

For a node $I$ set
\[
m_k(I):=\frac1{|I|}\sum_{j\in I}b_k(j),\qquad \iota_k(I):=\frac1{|I|}\sum_{j\in I}|b_k(j)|.
\]
Since $m_k(I)=\tfrac12(m_k(I_L)+m_k(I_R))$ and $\delta_I(k)=\tfrac12(m_k(I_R)-m_k(I_L))$,
\[
m_k(I_R)=m_k(I)+\delta_I(k),\qquad m_k(I_L)=m_k(I)-\delta_I(k).
\]
Let $j$ be a uniform random leaf and follow the root-to-leaf path: with filtration $\mathcal F_\ell$ generated by $I_{L-\ell}(j)$ ($\ell=0,\dots,L$; each child chosen with probability $1/2$ given its parent), the process
\[
t\mapsto m_k\big(I_{L-\ell}(j)\big),\qquad m_k(I_0(j)):=b_k(j),
\]
is a martingale with conditionally symmetric increments $\pm\,\delta_{I_\sigma(j)}(k)$ whose magnitudes are predictable (the scale-$\ell$ node is known before its child is chosen). Also, by the triangle inequality, $|\delta_I(k)| \le \iota_k(I)$, and averaging \eqref{eq:ink} over $j\in I$ gives
\begin{equation}
\label{eq:inkp}
\sum_k \iota_k(I)\le 1,\quad\text{for every node }I.
\end{equation}
Define the \emph{path variance} 
\begin{equation}\label{eq:pathvariance}
    V_k(j):=\sum_{\ell=1}^L \delta_{I_\ell(j)}(k)^2 \;\le\; L,
\end{equation}
since each $|\delta|\le\iota\le 1$.
Thus, $\delta_I(k)$ measures
the amount by which the average changes when you descend from $I$ to one of its children and $V_k(j)$ measures how much $\beta_k$ oscillates along the root-to-leaf path to $j$.  

Next, we show that if a martingale accumulates large squared jump sizes along a path, then its final value cannot remain tiny, i.e., a large path variance yields a lower bound on $\E[|M_T|]$. The following was proven in \cite{Osekowski14} (similar results that suffice for our setup could be obtained from any of \cite[Theorem 8]{Burkholder66}, \cite[Theorem 3.1]{Burkholder73}, \cite{Bollobas80}, or \cite{Cox82}, albeit with a worse constant).

\begin{lemma}\label{lem:mart}
Let $(M_t)_{t=0}^T$ be a finite discrete-time martingale on a filtration $(\mathcal F_t)$, $M_0=\mu$ deterministic, with $M_{t+1}=M_t\pm d_{t+1}$ where the sign is conditionally symmetric given $\mathcal F_t$ and $d_{t+1}\ge 0$ is $\mathcal F_t$-measurable. Let $V=\sum_{t=1}^T d_t^2$. Then for every $v>0$,  $\E|M_T| \;\ge\; \frac{\sqrt v}{e}\,\Pr[V\ge v]$.
\end{lemma}

Applied to path variances, we obtain the following corollary to upper bound the number of dictionary rows with large oscillation; since $B$'s columns are normalized, this cannot be too many. 
\begin{corollary}\label{cor:level}
For every $v>0$,
\[
\E_{j\sim [U]} \ \#\{k:\ V_k(j)\ge v\}\;\le\;\frac{e}{\sqrt v}.
\]
\end{corollary}

\begin{proof}
Fix $k$ and apply Lemma~\ref{lem:mart} to the mean martingale started at the root: $T=L$, $\mu=m_k([U])$, terminal value $M_L=b_k(j)$, path variance $V=V_k(j)$. Lemma~\ref{lem:mart} gives
\[
\Pr_j[V_k(j)\ge v]\;\le\;\frac{e}{\sqrt v}\,\E_j|b_k(j)|\;=\;\frac{e}{\sqrt v}\,\iota_k([U]).
\]
Sum over $k$ and use \eqref{eq:inkp} with the full interval $[U]$. 
\end{proof}

\subsection{Fresh direction and fresh distances}\label{sec:fresh}

The goal is to show that for every dyadic node $I$, the direction $\delta_I$ contains some component that cannot be explained by all same-or-coarser dyadic directions; this unexplained component is ``fresh component'' $\rho_I$ defined next; here, $D_I$ is the ``fresh distance''.  This will let us orthonormalize (similar to Gram--Schmidt) over the root-to-leaf path in the dyadic tree.   For a node $I$ at scale $\ell$, 
\[
E_I:=\spann\{\delta_{I'}:\ I' \text{ a node with } |I'|\ge|I|,\ I'\ne I\}\ \subseteq\ \R^m,
\]
\[
\rho_I:=\delta_I-\Pi_{E_I}\delta_I\ \ (\text{orthogonal projection}),\qquad D_I:=\|\rho_I\|_2.
\]
\begin{lemma}[Affinity of projections; $D_I>0$]\label{prop:a}
Write $\Pi_{E_I}\delta_I=\sum_{I'} c_{I'}\delta_{I'}$ (a finite linear combination over nodes $I'$ with $|I'|\ge|I|$, $I'\ne I$; fix one such representation per node $I$ once and for all) and define
$P_I(i):=\sum_{I'} c_{I'}\, g_{I'}(i)$.
Then $P_I$ restricted to $I$ is the restriction of an affine function on $[b,b+2^\ell]$, and $D_I>0$ for every node $I$ at scale $6\le \ell\le L$; hence $u_I:=\rho_I/D_I$ is a unit vector.
\end{lemma}
\begin{proof}
In the dyadic tree, a node $I'$ with $|I'|\ge |I|$, $I'\ne I$, is either disjoint from $I$ or strictly contains $I$. If $I'$ is disjoint from $I$ (in particular every same-scale $I'\ne I$), the canonical extension of $g_{I'}$ vanishes identically on $[b,b+2^\ell]$, hence is affine there. If $I'$ strictly contains $I$, its scale is $\ell'>\ell$, so all kinks of the canonical extension of $g_{I'}$ lie at multiples of $2^{\ell'-1}$, hence at multiples of $2^\ell$; the open interval $(b,b+2^\ell)$ contains no multiple of $2^\ell$, so $g_{I'}$ is affine on $[b,b+2^\ell]$. Hence $P_I$ is affine on $[b,b+2^\ell]$.

If $D_I=0$, then $\delta_I=\Pi_{E_I}\delta_I=\sum c_{I'}\delta_{I'}$, and \eqref{eq:I1} applied to every row $i$ gives
\[
g_I(i)=\langle a_i,\delta_I\rangle=\sum_{I'}c_{I'}\langle a_i,\delta_{I'}\rangle=\sum_{I'}c_{I'}g_{I'}(i)=P_I(i)\quad\text{for all }i,
\]
contradicting Lemma~\ref{lem:kink}, which forces at least $|I|/64\ge 1$ points of $I$ with $|g_I-P_I|\ge\tfrac1{16}$.
\end{proof}
Since $D_I > 0$, there is fresh information in every $\delta_I$.  The next lemma shows that the normalized fresh directions are orthogonal along a root-to-leaf path in the dyadic tree.
\begin{lemma}\label{lem:b}
Fix a row index $i$ and let 
$u_\ell:=u_{I_\ell(i)},
D_\ell = D_{I_\ell(i)}$ for $6\le \ell \le L$. Then $\{u_\ell\}_{\ell=6}^L$ is an orthonormal system in $\R^m$.
\end{lemma}
\begin{proof}
For $\ell<\ell'$: $u_{\ell'}=\rho_{I_{\ell'}(i)}/D_{\ell'}$ lies in $\spann\{\delta_{I'}: |I'|\ge 2^{\ell'}\}$; every such $I'$ has $|I'|>2^\ell$, hence $I'\ne I_\ell(i)$, so this span is contained in $E_{I_\ell(i)}$, to which $u_\ell$ is orthogonal by construction. So $\langle u_\ell,u_{\ell'}\rangle=0$, and each $u_\ell$ is a unit vector.
\end{proof}
Now that we have turned non-affine information into orthogonal Euclidean directions, we next show how this can be used to force a
large row norm in $A$.  
\begin{lemma}\label{lem:c}
For every row $i$ and node $I$ at scale $\ell\ge 6$,
$\langle a_i,u_I\rangle=\big(g_I(i)-P_I(i)\big)/D_I$,
and, with
\[
R(i):=\Big\{\ell\in[6,L]:\ \big|g_{I_\ell(i)}(i)-P_{I_\ell(i)}(i)\big|\ge \tfrac1{16}\Big\},
\]
we have
\begin{equation}\label{eq:41}
\|a_i\|_2^2\;\ge\;\sum_{\ell=6}^L\langle a_i,u_\ell\rangle^2\;\ge\;\frac1{256}\sum_{\ell\in R(i)} D_{I_\ell(i)}^{-2}.
\end{equation}
Moreover $\E_i|R(i)|\ge L/5$ for $L \ge 10$.%\frac{L-5}{64}\ge\frac{L}{65}$ for $L\ge 325$.
\end{lemma}

\begin{proof}
Using \eqref{eq:I1} for $\delta_I$ and for each $\delta_{I'}$ in the fixed representation of $\Pi_{E_I}\delta_I$:
\[
\langle a_i,u_I\rangle=\frac{\langle a_i,\delta_I\rangle-\langle a_i,\Pi_{E_I}\delta_I\rangle}{D_I}=\frac{g_I(i)-P_I(i)}{D_I}.
\]
(The value $\langle a_i,\Pi_{E_I}\delta_I\rangle$ is independent of the representation; $P_I$ is the fixed one.) By Lemma~\ref{lem:kink} applied to the affine function $P_I$: for each node $I$ at scale $\ell\ge 6$, at least $7|I|/16$ of the $i\in I$ satisfy $|g_I(i)-P_I(i)|\ge\tfrac1{16}$. Bessel's inequality for the orthonormal system of Lemma~\ref{lem:b} gives \eqref{eq:41}. Since scale-$\ell$ nodes partition $[U]$ and each contributes a $\ge \tfrac{7}{16}$ fraction of its points, $\Pr_{i}[\ell\in R(i)]\ge \tfrac{7}{16}$ for each $\ell\in[6,L]$, hence $\E_i|R(i)|\ge\frac{7(L-5)}{16}\ge\frac{L}{5}$ for $L\ge 10$.
\end{proof}
Given this lower bound on $\|a_i\|_2^2$, we want many of the $D_{I_\ell(i)}$ to be small.  This is what we show next: 
along most root-to-leaf paths in the dyadic tree, the total amount of ``fresh distance'' is only polylogarithmic in $L$.  
Fix a uniform row index $i$; write $D_\ell:=D_{I_\ell(i)}$, $V_k:=V_k(i)$, and $P_k:=\sum_{\ell=6}^L u_\ell(k)^2$.

\begin{lemma}\label{lem:chain}
%Let $r:=\lceil 3\log_2 L\rceil$ and $\Lambda:=260(r+2)$. 
For all $i$ outside a set of measure at most $1/16$,
\begin{equation}\label{eq:52}
\sum_{\ell=6}^L D_\ell\ \le\ K\ := O(\log L).%\ e\sqrt 2\Lambda r+\sqrt{2e}\sqrt\Lambda\ =\ O\big(\log^2 L\big).
\end{equation}
\end{lemma}

\begin{proof}
Since $u_\ell\perp \Pi_{E_{I_\ell(i)}}\delta_{I_\ell(i)}$,
\[
D_\ell=\langle u_\ell,\rho_{I_\ell(i)}\rangle=\langle u_\ell,\delta_{I_\ell(i)}\rangle\ \le\ \sum_k |u_\ell(k)|\,|\delta_{I_\ell(i)}(k)|.
\]
Summing over $\ell\in[6,L]$ and applying Cauchy--Schwarz, 
\begin{equation}\label{eq:51}
\sum_{\ell=6}^L D_\ell
\leq \sum_k \sum_{\ell = 6}^L |u_\ell(k)|\,|\delta_{I_\ell(i)}(k)|
\leq
\sum_k \left( \sum_{\ell=6}^L u_\ell(k)^2 \right)^{1/2} \left( \sum_{\ell = 6}^L \delta_{I_\ell(i)}(k)^2 \right)^{1/2}
\le
\sum_k\sqrt{P_k\,V_k},
\end{equation}
where we used $\sum_{\ell=6}^L\delta_{I_\ell(i)}(k)^2\le V_k$.

We now apply \cref{lem:innerprod} suitably. First, observe that $\mathbf U$ with rows $u_6,\dots,u_L$ has orthonormal rows, so $\mathbf U^\top \mathbf U$ is an orthogonal projection of rank $L-5$; thus $P_k=(\mathbf U^\top \mathbf U)_{kk}\in[0,1]$ and $\sum_k P_k=L-5\le L$. Next, by \eqref{eq:pathvariance}, we know $V_k \leq L$. 

Let random variables $X,Y \in R^m$ be defined by $X_k = \sqrt{P_k}$ and $Y_k = \sqrt{V_k}$. Combining the above with \cref{cor:level}, $X, Y$ satisfy the conditions of \cref{lem:innerprod}. Thus, we must have 
$$\E[\langle X, Y \rangle] = O(\log L).$$

The claim now follows from a Markov inequality. 
\end{proof}

\subsection{Final steps}

Since $0\le|R(i)|\le L$ always and $\E_i|R(i)|\ge L/5$ (Lemma~\ref{lem:c}), reverse Markov gives
\[
\Pr_i\Big[|R(i)|\ge \tfrac{L}{10}\Big]\ \ge\ \frac{L/5-L/10}{L-L/8} \ =\ \frac1{9}.
\]
Combining with \eqref{eq:52}, which fails on measure $\le 1/16$: a set of rows $i$ of measure $\ge 1/9 -1/16=7/144$ satisfies \emph{both} $|R(i)|\ge L/10$ and $\sum_{\ell=6}^L D_\ell\le K$. Fix such a row $i$. By, \cref{lem:jensen}, we  get
$$\sum_{\ell \in R(i)} \frac{1}{D_{I_\ell(i)}^2} \geq \frac{L^3}{1000 K^2}.$$

\ignore{

Now, for such a $i$, by Jensen's inequality, we have
$$\sum_{\ell \in R(i)} \frac{1}{D_{I_\ell(i)}^2}$$

At most $K/\lambda^\ast$ scales $\ell\in[6,L]$ have $D_\ell>\lambda^\ast$; choose $\lambda^\ast:=260K/L$, so at most $L/260$ scales are excluded, leaving at least $\frac{L}{130}-\frac{L}{260}=\frac{L}{260}$ scales $\ell\in R(i)$ with $D_\ell\le\lambda^\ast$.} 

By \eqref{eq:41}:
\[
\|a_i\|_2^2\ \ge\ \frac1{256000}\cdot \frac{L^3}{K^2}.
\]
With $K=O(\log L)$ this gives, for $L\ge 10$,
$\max_i \|a_i\|_2 \ge \Omega\big(L^{3/2}/\log L\big)$.

\bibliography{refs}
\bibliographystyle{alpha}

\appendix
\section{Correlation analysis for $\csk$}\label{sec:correlationfull}

The following is a restatement of \cref{thm:cov-mains} with explicit constants:
\begin{theorem}[Correlation bound for median-of-5 $\countsketch$]\label{thm:cov-main}
Consider $\csk$ with $\ell=5$ rows and $r$ buckets per row, with the median decoder.
Then for any $u \neq v$, and errors as defined in $\cref{eq:csk3}$, 
\[
\big|\E[\err(u) \cdot \err(v)]\big|
\;\le\;
\frac{70}{r^3}\,\|x\|_1^2.
\]
In particular, setting $r=\lceil 1/\varepsilon\rceil$ gives
$|\E[\err(u) \cdot \err(v)]|\le 70\varepsilon^3 \|x\|_1^2$.
\end{theorem}
We here provide the proof of \cref{thm:cov-mains}, together with details on how to derandomize it in low memory. We use the same setup as in \cref{sec:correlation}.

Fix distinct $i,j\in[U]$. For each row $a\in[5]$, call row $a$ \emph{bad}
if $h_a(i)=h_a(j)$, and let
\[
K=\bigl|\{a\in[5]:h_a(i)=h_a(j)\}\bigr|,
\]
be the number of bad rows.  Thus $K\sim\operatorname{Bin}(5,1/r)$.

We first state a simple fact about medians. 
\begin{lemma}
\label{lem:median}
Suppose
$z_1,\ldots,z_5\in\mathbb R$ and at most $k\le 2$ of the five
indices are bad. Let $G$ be the set of remaining good indices,
and let $R_{3-k}$ be the $(3-k)$th largest value among
$\{|z_a|:a\in G\}$. Then
$\left|\median(z_1,\ldots,z_5)\right|
\le R_{3-k}.$
\end{lemma}
\begin{proof}
If the median were larger than $R_{3-k}$, then at least three of the
five $z_a$'s would be larger than $R_{3-k}$. Since at most $k$ of these
indices are bad, at least $3-k$ good indices would have absolute value
larger than $R_{3-k}$, contradicting its definition. The case
where the median is smaller than $-R_{3-k}$ is similar.
\end{proof}

We will also use two elementary estimates for a good row. 
\begin{lemma}
\label{lem:goodrow}
We have
\[
(i) \E\left[|Y_{a,i}|\,\middle|\,h_a(i)\ne h_a(j)\right]
\le
\frac{\|x\|_1}{r} 
\quad \mbox{ and } \quad
(ii) \E\left[Y_{a,i}^2\,\middle|\,h_a(i)\ne h_a(j)\right]
\le
\frac{\|x\|_2^2}{r}.
\]
\end{lemma}
\begin{proof}
(i) follows from
\[
|Y_{a,i}|
\le
\sum_{u\ne i}|x_u| \cdot \mathbf 1 \{h_a(u)=h_a(i)\}.
\]
After conditioning on $h_a(i)\ne h_a(j)$, the term $u=j$ contributes zero,
and for every $u\notin\{i,j\}$,
\[
\Pr[h_a(u)=h_a(i)\mid h_a(i)\ne h_a(j)]=\frac1r.
\]
For (ii), since $h_a(i)\ne h_a(j)$,
the item $j$ is not in the bucket of $i$, and hence
\[
Y_{a,i}
=
\sum_{u\notin\{i,j\}:\,h_a(u)=h_a(i)}
\sigma_{a,i}\sigma_{a,u}x_u .
\]
Taking expectation over the signs kills the cross terms, giving
\[
\E_{\sigma}[Y_{a,i}^2\mid h_a]
=
\sum_{u\notin\{i,j\}:\,h_a(u)=h_a(i)} x_u^2.
\]
Averaging over the hashes, still conditioned on $h_a(i)\ne h_a(j)$, gives
(ii). The same estimates hold with $i$ and $j$ interchanged.
\end{proof}
Resuming the proof of \cref{thm:cov-mains}, we now do a case analysis on the values of $K$.

(i) $K=0.$  Then, $\E[\err(i) \cdot \err(j)\mid K=0]=0.$  Indeed, condition on all hash values and on all signs except
$\sigma_{1,i},\ldots,\sigma_{5,i}$. Since no row has $h_a(i)=h_a(j)$,
the quantity $\err(j)$ is fixed under this conditioning. On
the other hand, flipping all five signs
$\sigma_{a,i}\mapsto-\sigma_{a,i}$ flips all five values $Y_{a,i}$, and
therefore flips $\err(i)$. By symmetry of the signs, the
conditional expectation is zero.

(ii) $K=1$.  Condition on the identity of the unique bad row; the remaining four rows are good. By 
\cref{lem:median}, $|\err(i)|\le R_2$,
where $R_2$ is the second largest among the four random variables
$|Y_{a,i}|$ coming from the good rows. For nonnegative numbers
$X_1,\ldots,X_4$, the square of the second largest is at most
\[
\sum_{a<b}X_aX_b.
\]
Using independence across rows and \cref{lem:goodrow}(i), we get
\[
\E[\err(i)^2\mid K=1]
\le
\binom42\left(\frac{\|x\|_1}{r}\right)^2
=
6\frac{\|x\|_1^2}{r^2}.
\]
The same bound holds for $\err(j)$. Hence by Cauchy--Schwarz and since $\Pr[K=1]\le \frac5r$, 
\[
\left|
\E[
\err(i) \cdot \err(j) \cdot \mathbf 1 \{K=1\}
]
\right|
= 
\left|
\E[\err(i) \cdot \err(j)\mid K=1]
\right| \cdot \Pr[K = 1]
\le
6\frac{\|x\|_1^2}{r^2} \cdot \frac5r
= 30\frac{\|x\|_1^2}{r^3}.
\tag{3}
\]

(iii) $K = 2$.  Condition on the identities of the two bad rows; the remaining three rows are good. By \cref{lem:median}, 
$|\err(i)| \le
\max_{a\in G}|Y_{a,i}|$,
where $G$ is the set of good rows. Therefore
\[
\err(i)^2
\le
\sum_{a\in G}Y_{a,i}^2.
\]
Using \cref{lem:goodrow}(ii),
\[
\E[\err(i)^2\mid K=2]
\le
3\frac{\|x\|_2^2}{r}
\le
3\frac{\|x\|_1^2}{r}.
\]
The same bound holds for $\err(j)$.  Hence  by
Cauchy--Schwarz and since $\Pr[K=2]\le \binom52\frac1{r^2}=\frac{10}{r^2}$, 
\[
\left|
\E[
\err(i) \cdot \err(j) \cdot \mathbf 1 \{K=2\} 
\right|
= 
\left|
\E[\err(i) \cdot \err(j)\mid K=2]
\right| 
 \cdot \Pr[K=2]
\leq 
3\frac{\|x\|_1^2}{r} \cdot \frac{10}{r^2}
= 
30\frac{\|x\|_1^2}{r^3}.
\tag{4}
\]

(iv) $K \geq 3$.  On every outcome,
\[
|\err(i)|\le \|x\|_1
\qquad\text{and}\qquad
|\err(j)|\le \|x\|_1,
\]
and since 
$\Pr[K\ge 3]\le \binom53\frac1{r^3}=\frac{10}{r^3}$, 
\[
\left|
\E[
\err(i) \cdot \err(j) \cdot \mathbf 1 \{K\ge 3\} 
]
\right|
\le
10\frac{\|x\|_1^2}{r^3}.
\tag{5}
\]

Combining the four cases, the $K=0$ contribution is zero, while
(3), (4), and (5) give
\[
\left|
\E[\err(i) \cdot \err(j)]
\right|
\le
(30+30+10)\frac{\|x\|_1^2}{r^3}
=
70\frac{\|x\|_1^2}{r^3}.
\]
This proves the theorem.

\begin{remark} 
Note the proof breaks down for $\ell=3$ for the following reason.  The key step in the $K=1$ case of the proof is that, after removing the unique bad row, four good
rows remain. The median is then controlled by the second largest absolute
value among these four good rows. For $\ell = 3$, the analogous situation with $K=1$ leaves only two good rows. The median can then only be bounded by the maximum of the two
good-row magnitudes. This leads to a bound using the second moment $\frac{\|x\|_2^2}{r}$, rather than the square of the first moment.  Hence, the $K=1$
contribution is only controlled at scale
$O(\|x\|_1^2/r^2)$ 
and not at scale $O(\|x\|_1^2/r^3)$.

This is not only an artifact of the analysis. Consider the two-sparse vector supported on $\{i,j\}$ with equal magnitudes. One can check that with $\ell=3$ the median can become nonzero as soon as $K\ge 2$, which occurs with probability $\Theta(1/r^2)$, yielding a covariance scale $\Theta(\|x\|_1^2/r^2)$ rather than $\Theta(\|x\|_1^2/r^3)$.
Thus $\ell\ge 5$ is the first odd choice compatible with an $\varepsilon^3$-type target when $r=\Theta(1/\varepsilon)$. 
\end{remark}

\subsection{Derandomized Correlation Bound for $\csk$}\label{sec:derandomize}

We now explain how to derandomize the proof of \Cref{thm:cov-main}. The proof is written assuming fully independent hash functions and signs.
The same argument for the $K=1$, $K=2$, and $K\ge 3$ cases only needs
limited independence.

First, the collision tail bounds
\[
\Pr[K\ge t]\le {5\choose t}r^{-t}, \qquad t=1,2,3,
\]
only require independence across the five rows and the guarantee that, in
each row, $\Pr[h_a(u)=h_a(v)]=1/r$. 
Thus pairwise independence of each hash function $h_a$ would suffice for
this part.

Second, \Cref{lem:median} is deterministic. The $K=1$ part of the proof uses
\Cref{lem:goodrow}(i).  This estimate only uses that, for every $w\notin\{u,v\}$,
\[
\Pr[h_a(w)=h_a(u)\mid h_a(u)\ne h_a(v)]=1/r,
\]
which follows from $3$-wise independence of $h_a$.

Third, the $K=2$ part of the proof uses \Cref{lem:goodrow}(ii).  
For this estimate it is enough that $h_a$ is $3$-wise independent and that
the row-sign vector $\sigma_a$ is pairwise independent and independent of
$h_a$, since the proof only needs the cross terms in
$\mathbb E_\sigma[Y_{a,u}^2\mid h_a]$ to vanish.

The only part not covered by this limited independence argument
is the exact $K=0$ sign-flip cancellation. 
We now show how to handle the $K=0$ case if the $\sigma$ random variables come from a pseudorandom generator fooling intersections of $O(1)$ halfspaces \cite{GopalanOWZ10,GopalanKM18}. Note then there are two constraints on $\sigma$ (pairwise independence, and fooling intersections of halfspaces); we finally give a soft argument showing a simple way to achieve both simultaneously.

%Several steps in the proof of \cref{thm:cov-mains} above do not require full independence:
%\begin{itemize}
%  \item Lemma~\ref{lem:goodrow-L1} and Lemma~\ref{lem:goodrow-L2} require that for any $u\notin\{i,j\}$,
%  $h_a(u)$ is independent of $(h_a(i),h_a(j))$; $3$-wise independence across keys suffices.
%  \item Lemma~\ref{lem:goodrow-L2} uses that cross terms vanish under sign expectation; $4$-wise independence of signs is more than sufficient.
%\end{itemize}

%By following these steps carefully, and using pseudorandom generators for functions of halfspaces, we show the following. 

\begin{definition}
A \emph{halfspace} is a function $f:\{-1,1\}^m \rightarrow \R$ that can be expressed as either
$$
f(x) = \mathbf{1}\{\langle w,x\rangle \ge \theta \}, \text{ or}\qquad f(x) = \mathbf{1}\{\langle w,x\rangle > \theta \}
$$
for some $w\in\R^m, \theta\in\R$. We say a function $F:\{-1,1\}^m\rightarrow\R$ is an \emph{intersection of $k$ halfspaces} if it can be expressed as $F(x) = \prod_{i=1}^k f_i(x)$ for some halfspaces $f_1,\ldots,f_k$.

We then say a function $\mathcal G:\{0,1\}^s\rightarrow\{-1,1\}^m$ is a \emph{$\gamma$-pseudorandom generator against the intersection of $k$ halfspaces} (or \emph{$\gamma$-PRG}) if for any such intersection $F$,
$$
\left|\E_{x\in\{-1,1\}^m} F(x) - \E_{s\in\{0,1\}^s} F(\mathcal G(s)) \right| \le \gamma
$$
\end{definition}

\begin{theorem}[\cite{GopalanOWZ10}]\label{thm:prgs}
For any integer $m>0$ and $0<\gamma<1$, and any fixed constant integer $k>0$, there exists a $\gamma$-PRG against the intersection of $k$ halfspaces over $\{-1,1\}^m$ with seed length $s = O(\log m + \log^2(1/\gamma))$.
\end{theorem}

We now show that a PRG fooling the intersection of $4$ halfspaces suffices to handle the case $K=0$ of the proof.

\begin{lemma}\label{lem:derandomize-K0}
If $K=0$ (i.e., $h_a(u)\neq h_a(v)$ for all $a\in[5]$), the $\sigma_a$ are independently chosen across $a\in[5]$, and each $\sigma_a$ is itself the output of a PRG that $\gamma$-fools the intersection of $4$ halfspaces, then
\[
|\E[\err(u) \cdot \err(v)\mid K=0]| \le 5\cdot 2^{20}\gamma\|x\|_1^2 .
\]
\end{lemma}
\begin{proof}
Write $n:=\|x\|_1$.  For any real numbers $\alpha,\beta$ with
$|\alpha|,|\beta|\le n$, we have the identity
\[
\alpha\beta
=
\int_0^n\int_0^n
\Big(
  \mathbf 1\{\alpha>s,\beta>t\}
+ \mathbf 1\{\alpha<-s,\beta<-t\}
- \mathbf 1\{\alpha>s,\beta<-t\}
- \mathbf 1\{\alpha<-s,\beta>t\}
\Big)\,ds\,dt .
\]
Since $|\err(w)| \leq \|x\|_1$ for every coordinate $w$, we have $|\err(u)|,|\err(v)|\le n$, so it suffices to compare,
for each fixed $s,t\in[0,n]$, the probabilities of the four events
appearing in the integrand.

We treat the event
\[
\{\err(u)>s,\err(v)>t\};
\]
the other three events are identical.  Since
$\err(u)=\median(Y_{1,u},\ldots,Y_{5,u})$, the event $\err(u)>s$
is the disjoint union over all subsets $S\subseteq[5]$ with
$|S|\ge 3$ of the events
\[
\bigcap_{a\in S}\{Y_{a,u}>s\}
\cap
\bigcap_{a\notin S}\{Y_{a,u}\le s\}.
\]
There are
\[
\binom53+\binom54+\binom55=16
\]
such choices of $S$.  Similarly, $\err(v)>t$ is the disjoint union
over subsets $T\subseteq[5]$ with $|T|\ge3$ of the analogous events.
Therefore $\{\err(u)>s,\err(v)>t\}$ is a disjoint union of at most
$16^2=2^8$ events indexed by pairs $(S,T)$.

Fix such a pair $(S,T)$.  The corresponding event factors across rows:
\[
\prod_{a=1}^5 F_a(\sigma_a),
\]
where $F_a$ is the indicator of the two row-local constraints imposed
on $Y_{a,u}$ and $Y_{a,v}$.  We now show that each such $F_a$ is a sum
of at most $4$ indicators, each of which is an intersection of at most
$4$ halfspaces in the single row-sign vector $\sigma_a$.

Fix a row $a$.  Since we have conditioned on $K=0$, we have
$h_a(u)\neq h_a(v)$.  Hence $v$ is not in the bucket used to estimate
$u$, and $u$ is not in the bucket used to estimate $v$.  Define
\[
L_{a,u}
:=
\sum_{\substack{w\neq u\\ h_a(w)=h_a(u)}}
   \sigma_{a,w}x_w,
\qquad
L_{a,v}
:=
\sum_{\substack{w\neq v\\ h_a(w)=h_a(v)}}
   \sigma_{a,w}x_w.
\]
Then
\[
Y_{a,u}=\sigma_{a,u}L_{a,u},
\qquad
Y_{a,v}=\sigma_{a,v}L_{a,v}.
\]
For example,
\[
\{Y_{a,u}>s\}
=
\big(\{\sigma_{a,u}>0\}\cap\{L_{a,u}>s\}\big)
\;\dot\cup\;
\big(\{\sigma_{a,u}<0\}\cap\{L_{a,u}<-s\}\big),
\]
a disjoint union of two events, each an intersection of two halfspaces
in $\sigma_a$.  The same statement holds for the events
$\{Y_{a,u}\le s\}$, $\{Y_{a,u}< -s\}$, and $\{Y_{a,u}\ge -s\}$, and
similarly for $v$.  Thus the conjunction of the row-local $u$-constraint
and the row-local $v$-constraint is a disjoint union of at most $4$
events, each an intersection of at most $4$ halfspaces.

Consequently, for the fixed pair $(S,T)$, the corresponding global
event is a disjoint union of at most
\[
4^5=2^{10}
\]
events of the following product form:
\[
\prod_{a=1}^5 f_a(\sigma_a),
\]
where each $f_a$ is the indicator of an intersection of at most $4$
halfspaces in the row vector $\sigma_a$.

We now compare the expectation of one such product under independent
row-wise PRG signs and under independent uniform signs.  Let
\[
p_a:=\E_{\sigma\sim\mathcal D} f_a(\sigma),
\qquad
q_a:=\E_{\sigma\sim\mathcal U} f_a(\sigma).
\]
Since $\mathcal D$ $\gamma$-fools intersections of $4$ halfspaces,
\[
|p_a-q_a|\le \gamma
\qquad\text{for every }a\in[5].
\]
Using independence across rows,
\[
\E_{\sigma_1,\ldots,\sigma_5\sim\mathcal D}
 \prod_{a=1}^5 f_a(\sigma_a)
=
\prod_{a=1}^5 p_a,
\]
and similarly the uniform expectation is $\prod_{a=1}^5 q_a$.
Since $p_a,q_a\in[0,1]$,
\[
\left|
\prod_{a=1}^5 p_a-\prod_{a=1}^5 q_a
\right|
\le
\sum_{a=1}^5 |p_a-q_a|
\le
5\gamma .
\]

For each fixed $s,t$, each of the four integrand events is a disjoint
union of at most $2^8$ median-pattern events, and each such event
expands into at most $2^{10}$ product terms.  Therefore the total
additive error in the integrand is at most
\[
4\cdot 2^8\cdot 2^{10}\cdot 5\gamma
=
5\cdot 2^{20}\gamma .
\]
Integrating over $[0,n]^2$ gives
\[
\left|
\E_{\mathcal D^{\otimes 5}}[\err(u)\err(v)\mid h]
-
\E_{\mathcal U^{\otimes 5}}[\err(u)\err(v)\mid h]
\right|
\le
5\cdot 2^{20}\gamma n^2 .
\]

It remains only to note that under fully uniform, independent row signs,
the conditional expectation is zero for every fixed hash realization
$h$ with $K=0$.  Indeed, condition on all signs except
$\{\sigma_{a,u}\}_{a=1}^5$.  Since $K=0$, $\err(v)$ does not depend on
these five signs, while flipping all five signs
$\sigma_{a,u}\mapsto-\sigma_{a,u}$ flips every $Y_{a,u}$ and therefore
flips $\err(u)$.  Hence
\[
\E_{\mathcal U^{\otimes 5}}[\err(u)\err(v)\mid h]=0.
\]
This proves the claimed bound for every fixed $h$ with $K=0$, and
averaging over such $h$ proves the lemma.
\end{proof}

To conclude, we can set $\gamma = 2^{-20}/(5r^3)$ to ensure the $K=0$ case does not introduce more than $O(\|x\|_1^2/r^3)$ into the correlation bound (which is already dominated by the $K=1$ and $K=2$ cases). We need $\sigma$ to both be a PRG fooling intersections of halfspaces, and be pairwise independent, which can be achieved by generating a PRG output and a pairwise independent $\{-1,1\}^m$ independently, then multiplying them bitwise (see \cref{lem:bitwise} below). The total seed length required, by \cref{thm:prgs} is $O(\log m + \log^2(1/\eps))$ bits, which is no more than $O(\log(1/\eps))$ words, dominated by the rest of the algorithm which always consumes $\Omega(1/\eps)$ words. The space to store five $3$-wise independent $h_a$ functions is also only $O(1)$ words.

\begin{lemma}\label{lem:bitwise}
Let $\sigma^{(\text{PRG})} \in \{-1, 1\}^m$ be a random vector generated by a pseudorandom generator (PRG) that $\gamma$-fools the intersection of $k$ halfspaces. Let $\sigma^{(\text{PI})} \in \{-1, 1\}^m$ be a random vector drawn from a pairwise independent distribution over $\{-1, 1\}^m$, drawn independently from $\sigma^{(\text{PRG})}$. 
Define $\sigma \in \{-1, 1\}^m$ via the bitwise (Hadamard) product $\sigma = \sigma^{(\text{PRG})} \circ \sigma^{(\text{PI})}$.

Then  $\sigma$ both (1) is pairwise independent, and (2) $\gamma$-fools the intersection of $k$ halfspaces.
\end{lemma}
\begin{proof}
    We prove the two properties separately.

    First we show $\sigma$ is pairwise independent. For this, it is equivalent to show that both the marginal expectation of every coordinate is $0$, and that the covariance of any pair of distinct coordinates is $0$.

Because $\sigma^{(\text{PI})}$ is pairwise independent over $\{-1, 1\}^m$, we already know that for its coordinates, $\E\sigma^{(\text{PI})}_i = 0$ for all $i \in [m]$, and
    $\E[\sigma^{(\text{PI})}_i \sigma^{(\text{PI})}_j] = 0$ for all $i \neq j \in [m]$.

For any coordinate $i \in [m]$, because $\sigma^{(\text{PRG})}$ and $\sigma^{(\text{PI})}$ are drawn independently, 
\begin{align*}
\E\sigma_i &= \E[\sigma^{(\text{PRG})}_i \cdot \sigma^{(\text{PI})}_i] \\
&= (\E\sigma^{(\text{PRG})}_i) \cdot (\E\sigma^{(\text{PI})}_i) \\
&= (\E\sigma^{(\text{PRG})}_i) \cdot 0 = 0
\end{align*}

For any distinct pairs of coordinates $i \neq j \in [m]$:
\[
\E[\sigma_i \sigma_j] = \E\left[\left(\sigma^{(\text{PRG})}_i \sigma^{(\text{PI})}_i\right) \left(\sigma^{(\text{PRG})}_j \sigma^{(\text{PI})}_j\right)\right]
\]
By rearranging terms and utilizing the mutual independence of the underlying vectors:
\begin{align*}
\E[\sigma_i \sigma_j] &= \E\left[\sigma^{(\text{PRG})}_i \sigma^{(\text{PRG})}_j\right] \cdot \E\left[\sigma^{(\text{PI})}_i \sigma^{(\text{PI})}_j\right] \\
&= \E\left[\sigma^{(\text{PRG})}_i \sigma^{(\text{PRG})}_j\right] \cdot 0 = 0
\end{align*}

Thus we have $\E\sigma_i = 0$ for all $i$, and $\E[\sigma_i \sigma_j] = 0 = \E[\sigma_i]\E[\sigma_j]$ for all distinct $i \neq j$, implying the coordinates of $\sigma$ are pairwise independent.

Now we show $\sigma$ fools the intersection of halfspaces.
Recall that an intersection of $k$ halfspaces is a boolean function $F: \{-1, 1\}^m \to \{0, 1\}$ defined as:
\[
F(x) = \prod_{\ell=1}^k \mathbbm{1}\left(\langle w^{(\ell)}, x \rangle > \theta^{(\ell)}\right)
\]
where $w^{(\ell)} \in \mathbb{R}^m$ (one may also replace some of the $>$ with $\ge$ in the product).

We will prove that $\left| \E[F(\sigma)] - \E_u[F(u)] \right| \le \gamma$ for the uniform distribution $u \sim \{-1, 1\}^m$.

First, condition on $\sigma^{(\text{PI})}$.
We can write:
\[
\E_{\sigma}[F(\sigma)] = \E_{\sigma^{(\text{PI})}}  \E_{\sigma^{(\text{PRG})}} [ F(\sigma^{(\text{PRG})} \circ y) \mid y = \sigma^{(\text{PI})} ]
\]
Fix an arbitrary outcome vector $y \in \{-1, 1\}^m$ and define a modified test function $F_y(x) = F(x \circ y)$. Let us expand $F_y(x)$:
\[
F_y(x) = \prod_{\ell=1}^k \mathbbm{1}\left(\langle w^{(\ell)}, x \circ y \rangle > \theta^{(\ell)}\right)
\]
Notice that the inner product evaluates to:
\[
\langle w^{(\ell)}, x \circ y \rangle = \sum_{i=1}^m w^{(\ell)}_i x_i y_i = \sum_{i=1}^m (w^{(\ell)}_i y_i) x_i = \langle \tilde{w}^{(\ell)}, x \rangle
\]
where we have defined a new weight vector $\tilde{w}^{(\ell)} \in \mathbb{R}^m$ such that $\tilde{w}^{(\ell)}_i = w^{(\ell)}_i y_i$. Since $y_i \in \{-1, 1\}$, this merely flips the signs of a fixed subset of the weights. Consequently:
\[
F_y(x) = \prod_{\ell=1}^k \mathbbm{1}\left(\langle \tilde{w}^{(\ell)}, x \rangle > \theta^{(\ell)}\right)
\]
Thus $F_y(x)$ is itself exactly an intersection of $k$ halfspaces, and so for any fixed $y$
\[
\left| \E_{\sigma^{(\text{PRG})}}\left[F_y(\sigma^{(\text{PRG})})\right] - \E_{u \sim \{-1, 1\}^U}\left[F_y(u)\right] \right| \le \gamma
\]
Also note $u\circ y$ is uniform for $u$ uniform and any choice of $y$. Therefore
\[
\E_{u \sim \{-1, 1\}^m}[F(u \circ y)] = \E_{u \sim \{-1, 1\}^m}[F(u)] ,
\]
and substituting this back, we have that for every fixed outcome $y \in \{-1, 1\}^m$
\[
\left| \E_{\sigma^{(\text{PRG})}}\left[F(\sigma^{(\text{PRG})} \circ y)\right] - \E_u\left[F(u)\right] \right| \le \gamma ,
\]
and thus the inequality holds for an average $y$ (where $y$ is drawn as $\sigma^{(\text{PI})}$).
\end{proof}

%\begin{theorem}\label{th:cskderand}
 %   We can choose pseudorandom hash functions $h_a:[U] \rightarrow [r]$, and signs $\sigma_{a,u} \in \{1,-1\}$ with a total seed of length $O(r)$ to satisfy the same guarantees as in \cref{thm:cov-mains}. In addition, the signs and hash functions can be computed efficiently. 
%\end{theorem}
\section{Matrix Mechanism in Differential Privacy}
\label{app:matmech}

Recall that in the world of {\it differential privacy} (DP), one has a private database $x$ and wants to answer queries without leaking too much privacy. For example, think of $x$ as being a histogram of user types, with $x_u$ being the number of users of type $u$. Then, the privacy guarantee is that one wants that for ``neighboring'' datasets $x,x'$, the distributions of $\mathcal M(x)$ and $\mathcal M(x')$ are  similar in some precise quantitative sense\footnote{We intentionally avoid the precise mathematical formalism here, since it is unrelated to the setting of our work.}, where $\mathcal M$ is the (randomized) private mechanism. Here, ``neighboring'' means that $x,x'$ differ in just one user, i.e.,  $x - x' = \pm e_u$ for some standard basis vector $e_u$. In other words, the output distribution of the mechanism should not be too sensitive to the presence of any one individual user. Then, the matrix mechanism considers problems that can be represented as the data analyst wanting to learn $Mx$ for some matrix $M\in\R^{U\times U}$. The private mechanism then factorizes $M = AB$, releases $Bx + \eta$ for some noise vector $\eta$, and the analyst then computes its estimate of $Mx$ as $A(Bx+\eta)$. The distribution of the noise $\eta$ depends on the specific notion of privacy desired (e.g.,  pure- vs. approximate-DP), which yields different $\gamma$-norms characterizing the privacy--utility tradeoff.
    
For example, in the case of $\eps$-DP, to release a function $f(x)$, it is well known that one should have the $\eta_i$ be i.i.d.\ Laplace noise with scale parameter  $\Delta_1(f)/\eps$ for $\Delta_1(f)$ denoting the {\it $\ell_1$-sensitivity}  $\sup_{\|x-x'\|_1 \le 1}\|f(x) - f(x')\|_1$. For $(\eps,\delta)$-DP, one should add Gaussian noise with standard deviation $\Delta_2(f)\sqrt{C\log(1/\delta)}/\eps$ for $\Delta_2(f) := \sup_{\|x-x'\|_1 \le 1} \|f(x)-f(x')\|_2$ being the {\it $\ell_2$-sensitivity}. In the case of the matrix mechanism, since we would like to publish $f(x) = Bx$ to the analyst, the $\ell_1$- and $\ell_2$-sensitivities are $\|B\|_{1\rightarrow 1}$ and $\|B\|_{1\rightarrow 2}$, respectively. If we would then like to minimize $\|y - \hat{y} \|_\infty$, note the error depends on $\sup_i\left|\langle A_{i:}, \eta\rangle\right|$, where $A_{i:}$ is the $i$th row of $A$ and $\eta$ has i.i.d.\ entries that are either Laplace or Gaussian.  

For pure-DP, one uses Laplacian noise for $\eta$. For each $i$, let $e_i = A_{i:} \eta$ be the error. Then, we need a bound on $\max_i |e_i|$. Note that each $e_i$ is a sum of independent mean zero random variables, and thus we can compute its standard deviation, which is proportional to $\|A_{i:}\|_2 \cdot(\E\eta_1^2)^{1/2} = O(\|A_{i:}\|_2\cdot \|B\|_{1\rightarrow 1} / \varepsilon)$. Further, we can use Bernstein inequality for sub-exponential random variables to get a sharper bound on the tail behavior of $e_i$: By Corollary 2.9.2 in \cite{Vershynin_2018}, if we let $K = \|B\|_{1 \rightarrow 1}/\varepsilon$, then  

$$\pr[ |e_i| > t] \leq 2 \exp\left[- c \min\left(\frac{t^2}{K^2 \|A_{i:}\|_2^2}, \frac{t}{K \|A_{i:}\|_\infty}\right)\right]. $$

By a standard union bound argument\footnote{Setting 
$t = O(K) \cdot \left((\sqrt{\log U}) \max_i \|A_{i:}\|_2 + (\log U) \max_i \|A_{i:}\|_\infty\right)$, gives $\pr[\max_i |e_i| > t] \leq U \exp(-\Omega(\log U))$.}, we get that
$$\ex\left[\max_i |e_i|\right] = O(K) \cdot \left((\sqrt{\log U}) \max_i \|A_{i:}\|_2 + (\log U) \max_i \|A_{i:}\|_\infty\right).$$

Therefore, for pure-DP, the error from the Laplace mechanism applied to the factorization $M = AB$ is 

\begin{equation}O(1) \left(\frac{\|A\|_{2 \rightarrow \infty} \|B\|_{1 \rightarrow 1} \cdot \sqrt{\log U}}{\varepsilon} + \frac{\max_i \|A_{i:}\|_\infty \cdot \|B\|_{1 \rightarrow 1} \cdot (\log U)}{\varepsilon}\right). \label{eqn:dp-ub}
\end{equation}

The first term above is essentially the $\gamma_{2,1}(M)$ norm again: $\gamma_{2,1}(M) = \inf_{A,B:AB=M} \{\|A\|_{2\rightarrow\infty} \cdot \|B\|_{1 \rightarrow 1}\}$. Note that the $\gamma_{2,1}$-norm resembles the first term in our \cref{eqn:main-space}. Using the above argument with dyadic factorization yields a bound of $O((\log U)^2/\varepsilon)$ error for pure-DP \cite{JainRSS23}. 

%Furthermore, although each error term has standard deviation $\|A_{i:}\|_2 \|B\|_{1\rightarrow 1}/\varepsilon$, since $\eta$ is random, a naive union bound over all $i$ yields that the maximum error over all $i$ is bounded by $O((\|A\|_{2\rightarrow\infty}\|B\|_{1\rightarrow 1}\log U) / \varepsilon)$. In the case of $M$ being the lower triangular matrix of all ones and with the dyadic factorization, this yields a bound of $O((\log U)^{2.5} /\varepsilon)$ \cite{ChanSS10,DNPR10}. A more careful analysis that improves upon the naive union bound can tighten this to $O((\log U)^2/\varepsilon)$ \cite{JainRSS23}.

\end{document}